\documentclass[11pt,a4paper]{article}
\pdfoutput=1
\usepackage{jheppub}
\usepackage[ascii]{inputenc}
\usepackage{bbm,url,braket,microtype,amsthm,mathtools,mathrsfs,cleveref,graphicx,float}
\usepackage{lmodern}
\usepackage[T1]{fontenc}
\usepackage[USenglish]{babel}

\newtheorem{thm}{Theorem}\crefname{thm}{Theorem}{Theorems}
\newtheorem{lem}[thm]{Lemma}\crefname{lem}{Lemma}{Lemmas}
\newtheorem{prp}[thm]{Proposition}\crefname{prp}{Proposition}{Propositions}
\crefname{cor}{Corollary}{Corollaries}
\crefname{prb}{Problem}{Problems}
\crefname{dfn}{Definition}{Definitions}
\crefname{section}{Section}{Sections}
\crefname{appendix}{Appendix}{Appendices}
\DeclareMathOperator{\tr}{tr}
\DeclareMathOperator{\rank}{rk}
\DeclareMathOperator{\id}{id}

\newcommand{\CC}{\mathbb C}

\newcommand{\EE}{\mathbb E}

\newcommand{\C}{\mathbb{C}}

\newcommand{\norm}[1]{\lVert{#1}\rVert}
\newcommand{\abs}[1]{\lvert{#1}\rvert}
\newcommand{\ot}{\otimes}
\newcommand{\op}{\oplus}
\allowdisplaybreaks[4]

\title{Conditional Mutual Information of Bipartite Unitaries and Scrambling}
\author{Dawei Ding,}
\author{Patrick Hayden,}
\author{Michael Walter}
\affiliation{Stanford Institute for Theoretical Physics, Stanford University, Stanford, California 94305, USA}
\emailAdd{dding@stanford.edu}
\emailAdd{phayden@stanford.edu}
\emailAdd{michael.walter@stanford.edu}
\abstract{%
One way to diagnose chaos in bipartite unitary channels is via the tripartite information of the corresponding Choi state, which for certain choices of the subsystems reduces to the negative conditional mutual information (CMI).
We study this quantity from a quantum information-theoretic perspective to clarify its role in diagnosing scrambling.
When the CMI is zero, we find that the channel has a special normal form consisting of local channels between individual inputs and outputs.
However, we find that arbitrarily low CMI does not imply arbitrary proximity to a channel of this form, although it does imply a type of approximate recoverability of one of the inputs.
When the CMI is maximal, we find that the residual channel from an individual input to an individual output is completely depolarizing when the other input is maximally mixed.
However, we again find that this result is not robust.
We also extend some of these results to the multipartite case and to the case of Haar-random pure input states.
Finally, we look at the relationship between tripartite information and its R\'enyi-2 version which is directly related to out-of-time-order correlation functions.
In particular, we demonstrate an arbitrarily large gap between the two quantities.}

\begin{document}
\maketitle

\section{Introduction}

Recent research in quantum gravity has led to an interest in the scrambling and chaotic properties of many-body quantum systems~\cite{hayden2007black,sekino2008fast,shenker2014black,shenker2014multiple,kitaev2014hidden,roberts2015diagnosing,maldacena2015bound}.
The simplest model to consider is that of a unitary time evolution, $U_{AB\to CD}$, where $A$,$B$ and $C$,$D$ denote fixed bipartitions of past and future time slices of the quantum system, respectively.
Typically, $A=C$ and $B=D$, and we merely use different letters to denote the past and future timeslices, but we may also consider two different bipartitions if we want to compare the propagation between different subsystems.

For chaotic dynamics, we expect that the local degrees of freedom $A$,$B$ will get encoded nonlocally into $C,D$, i.e., \emph{scrambled}.
One way to formalize this intuition, proposed recently in~\cite{hosur2015chaos}, is to consider the \emph{Choi state} dual to $U$, which is commonly used in quantum information theory to study the properties of quantum channels~\cite{wilde2013quantum}. For the specific case of bipartite unitaries, the Choi states are used to study the capacity~\cite{bennett2003capacities,harrow2004coherent,harrow2010time,berry2007lower,harrow2011communication} and the cost of implementation~\cite{wakakuwa2015coding, wakakuwa2015cost}.
In the present context, this is the pure state defined by
\begin{equation}
\label{eq:choi}
	\rho_{ABCD} = U_{A'B'\to CD} (\Phi^+_{AA'} \ot \Phi^+_{BB'}) U_{A'B'\to CD}^\dagger,
\end{equation}
where $\Phi^+_{AA'}$ and $\Phi^+_{BB'}$ denote maximally entangled states (\cref{fig:bipartite}), and it allows us to study the past and future subsystems on equal footing.
For scrambling unitaries, we expect the local correlations, as measured by the mutual informations $I(A;C) := S(A)+S(C)-S(AC)$ and $I(A;D)$, to be suppressed, while $I(A;CD)$ is necessarily maximal by unitarity.
This suggests the \emph{tripartite information}
\[ I_3(A;C;D) := I(A;C) + I(A;D) - I(A;CD), \]
or more precisely $-I_3$, as a measure of scrambling in unitary quantum channels.
It is easy to verify that the tripartite information does not depend on the choice of three subsystems $A$,$B$,$C$ of the four-party pure quantum state $\rho_{ABCD}$.

\begin{figure}
  \centering
  \includegraphics[width=0.3\textwidth]{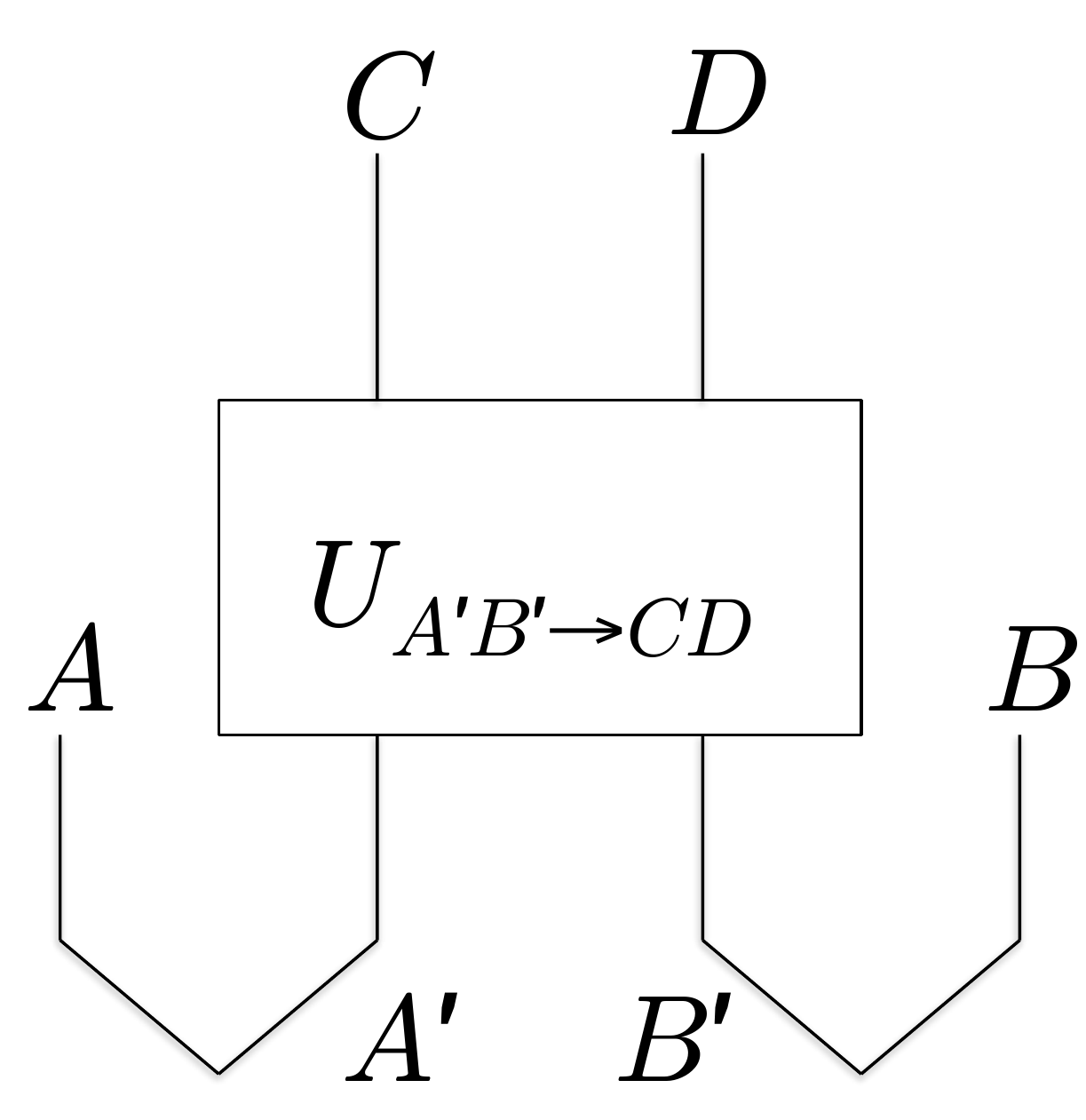}
  \caption{Choi state of a bipartite unitary $U$.}
  \label{fig:bipartite}
\end{figure}

The starting point to our investigations is the observation that unitarity implies that the reduced density matrices $\rho_{AB}$ and $\rho_{CD}$ of the Choi state are maximally mixed.
It follows that $I(A;B)=I(C;D)=0$ and hence the negative tripartite information reduces to
\begin{equation}
\label{eq:neg tri is cmi}
  -I_3 = I(A;B|C),
\end{equation}
where $I(A;B|C) = I(A;BC)-I(A;C)$ is the \emph{conditional mutual information} (CMI).%
\footnote{Likewise, $-I_3 = I(A;B|D) = I(C;D|A) = I(C;D|B)$. Note that other choices of subsystems might not reduce $-I_3$ to the CMI.}
In particular, the tripartite information is never positive as a consequence of the strong subadditivity of the von Neumann entropy:
\[ I_3 \leq 0. \]
This is true for an arbitrary unitary time evolution, whether chaotic or not, contrary to previous expectations~\cite{hosur2015chaos}.
Interestingly, $I_3\leq0$ is \emph{not} true for general quantum states, but it has recently been proved in a different context, namely as the consequence of the Ryu-Takayanagi formula in holographic systems~\cite{hayden2013holographic} (cf.~\cite{balasubramanian2014multiboundary,bao2015holographic}) and its tensor network models~\cite{pastawski2015holographic,hayden2016holographic}, where it can be interpreted as a consequence of the monogamy of entanglement~\cite{nezami2016multipartite}.
Whether there exists a deeper common reason for the negativity of $I_3$ associated to unitary transformations and the negativity of $I_3$ of a holographic state remains a tantalizing open question.

\begin{figure}
\centering
\includegraphics[width=.3\linewidth]{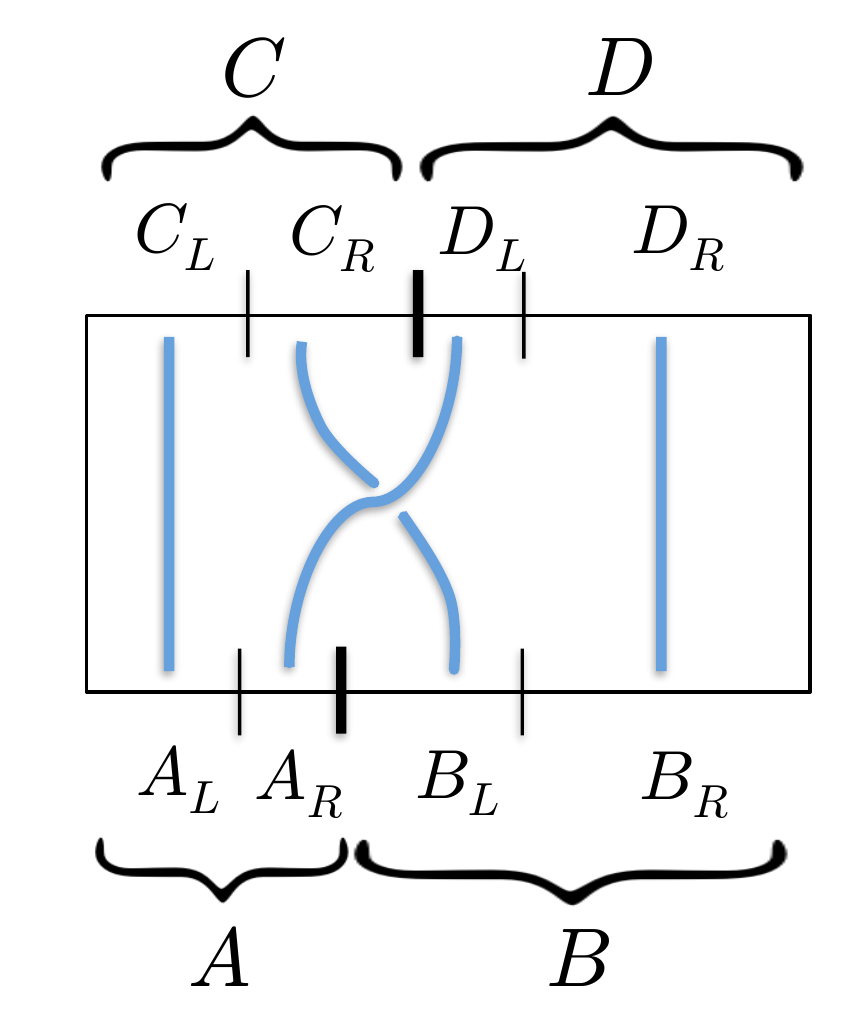}
\caption{Any bipartite unitary with $I_3=0$ is a `criss-cross channel' of the form~\eqref{eq:normal form}, routing the quantum information from the input to the output subsystems.}
\label{fig:normal form}
\end{figure}

\medskip

In this paper, we aim to clarify the meaning of the tripartite information from the perspective of quantum information theory, based on the connection established above.
We are particularly interested in the extreme cases, where the tripartite information attains its minimal or maximal values. 
We say that $U$ is \emph{minimally $I_3$-scrambling} if $I_3=0$ and \emph{maximally $I_3$-scrambling} if it attains its maximally negative value.

We start in \cref{sec:minimal} by considering the case of minimal $I_3$-scrambling.
Our first result shows that any such unitary has the following special form:
\begin{equation}
  U_{AB\to CD} = U_{A_L \to C_L} \ot U_{A_R \to D_L} \ot U_{B_L \to C_R} \ot U_{B_R\to D_R},
  \label{eq:normal form}
\end{equation}
for some decomposition $A=A_L \ot A_R$ and likewise for $B,C,D$ (see \cref{fig:normal form} for an illustration).
That is, the unitary can be decomposed into, in general, four smaller unitaries which locally route the quantum information between the input and output subsystems.
Such a `criss-cross channel' exactly matches our intuition of what a non-scrambling process should look like.
This result can also be interpreted as maximizing simultaneously achievable rates of communication between the input and output subsystems:
For example, we have that
\begin{equation}
\label{eq:rate local global}
  R_{A\to C} + R_{A\to D} = Q_{A\to CD},
\end{equation}
where we write $R_{A\to C}$ and $R_{A\to D}$ for the simultaneously achievable (one-shot, zero-error) quantum communication rates from $A$ to $C$ and $D$, respectively, and $Q_{A\to CD}$ for the quantum capacity from $A$ to $CD$, which by unitarity is always equal to $\log\abs A$, the Hilbert space dimension of $A$. Note that logarithms in this paper are base 2, in accordance to the convention in quantum information.
Lastly, our result can also be translated into a statement about the recoverability of the systems from partial information --- for the purposes of recovering the quantum information from input $A$ given output $D$, access to the other input subsystem $B$ does not help.

It is interesting to ask to what extent the above statements can be generalized to the case where $I_3\approx0$.
The latter result can be readily generalized to the approximate case using a recent result in quantum information theory~\cite{sutter2015universal}, which asserts that we can find a quantum operation $\mathcal R_{D\to BD}$ independent of the state at $A$ such that we can approximately recover $\rho_{ABD}$ from $\rho_{AD}$.
On the other hand, we show that~\eqref{eq:normal form} is not robust in the following, strongest possible sense: we explicitly construct a family of unitary quantum channels such that $I_3$ is arbitrarily close to zero, while their distance from any unitary of the form~\eqref{eq:normal form} is lower-bounded by a positive constant.
Our construction implies that any robust version of~\eqref{eq:normal form} must necessarily depend on the Hilbert space dimension.

From the perspective of quantum information theory, our results complement the nonrobustness result in \cite{ibinson2008robustness,christandl2012entanglement} that provide examples of tripartite states with vanishing conditional mutual information but non-vanishing trace distance to any quantum Markov chain state, that is, a state with a special normal form equivalent to having zero CMI.
Here, on the other hand, we find a tripartite state with vanishing CMI \emph{and} trace distance, but still with non-vanishing diamond norm to any quantum Markov chain state when the states are viewed as reduced Choi states of bipartite unitaries.
This provides further evidence for the nonrobustness of normal forms for quantum Markov chains.

\medskip

In \cref{sec:maximal} we then consider the other extreme case, where the tripartite information $I_3$ is maximally negative.
This can be achieved by, e.g., perfect tensors~\cite{pastawski2015holographic}, also known as absolutely maximally entangled states~\cite{goyeneche2015absolutely, helwig2012absolute}, such as those obtained by the random construction of~\cite{hayden2016holographic}.
Here, we give an explicit construction similar to that of~\cite{keet2010quantum} in the case $A=B=C=D$, which works in arbitrary odd dimensions. We also show that maximally scrambling unitaries do not exist if all the systems are qubits.

Now suppose that $U$ is maximally $I_3$-scrambling and, for concreteness, that the dimension of $A$ is the smallest among the four subsystems, so that $I_3 = -2\log\;\lvert A\rvert$.
Then the residual channels $\mathcal N_{A\to C}$ and $\mathcal N_{A\to D}$, obtained by fixing a maximally mixed state $\tau_B$ into $B$, applying the unitary, and tracing out either $D$ or $C$, are completely depolarizing.
\footnote{This is true only when the input on $B$ is fixed to be maximally mixed. In general, there may be some correlations between $A$ and $C$ or $D$.}
In other words, we cannot locally route any information from $A$ to $C$ or $D$,%
while we still have $R_{A\to CD}=\log\;\lvert A\rvert$ by unitarity.
This characterization nicely complements~\eqref{eq:normal form} and \eqref{eq:rate local global}.
It also complements the recovery interpretation: with only $D$, we can recover none of the information from $A$, but with $BD$ we can recover all of it.
However, we again find that we need to be cautious when generalizing this result to the approximate case:
We construct a unitary such that $I_3$ is arbitrarily close to being maximally negative, but whose residual channel $\mathcal N_{A\to C}$ is bounded away from the completely depolarizing channel.

\medskip

In \cref{sec:general}, we consider general values of $I_3$, again using the connection~\eqref{eq:neg tri is cmi} to the conditional mutual information.
The latter has an operational interpretation in the task of quantum state redistribution.
More precisely, given a quantum state $\rho_{ACD}$ with purification $\rho_{ABCD}$, if one party possesses $AC$ and another party $D$, the former can send $A$ to the latter using at an optimal rate of $\frac12 I(A;B \vert D)=-\frac12I_3$ qubits~\cite{devetak2008exact}.
This is intuitive: given that a strongly scrambling unitary will delocalize information from the inputs, we indeed expect that a larger number of qubits should be required to transfer systems. We show that this is consistent with our main results for minimal and maximal $I_3$-scrambling and give simple protocols that achieve the given qubit rate. Note that it is also possible to do similar analyses using other operational interpretations of CMI such as in the tasks of state deconstruction and conditional erasure~\cite{berta2016deconstruct}.

\medskip

An appealing feature of the tripartite information is that it is related to \emph{out-of-time-order (OTO) correlators}, an alternative diagnostic of chaos proposed to quantify the analog of the `butterfly effect' in black holes~\cite{shenker2014black}.
OTO correlators can also be measured in various physical systems~\cite{swingle2016measuring,yao2016interferometric}.
An OTO correlator of two local operators $\mathcal O_A$ and $\mathcal O_C$ is by definition an expectation value of the form
\[ \braket{\mathcal O_C(t) \mathcal O_A \mathcal O_C(t) \mathcal O_A}_\beta = \frac 1 Z \tr \bigl[ e^{-\beta H} \mathcal O_C(t) \mathcal O_A \mathcal O_C(t) \mathcal O_A \bigr], \]
where $U=e^{-iHt}$ is the time evolution operator and $\mathcal O_C(t) = U^\dagger \mathcal O_C U$.
We define the \emph{average OTO correlator} between $A$ and $C$, denoted $\lvert\braket{\mathcal O_C(t) \mathcal O_A \mathcal O_C(t) \mathcal O_A}_\beta\rvert$, by averaging the above over orthonormal bases of operators on $A$ and $C$.
In the infinite temperature limit, $\beta=0$, it is known that~\cite{hosur2015chaos}
\[
\lvert\braket{\mathcal O_C(t) \mathcal O_A \mathcal O_C(t) \mathcal O_A}_{\beta=0}\rvert
\times
\lvert\braket{\mathcal O_D(t) \mathcal O_A \mathcal O_D(t) \mathcal O_A}_{\beta=0}\rvert
\propto 2^{I_3^{(2)}}.
\]
Here, $I_3^{(2)} = S_2(A) + S_2(B) - S_2(AC) - S_2(AD)$ is a variant of the tripartite information%
\footnote{Note that we can similarly write $I_3(A;C;D) = -I(C;D|A) = S(A) + S(B) - S(AC) - S(AD)$.}
defined in terms of the R\'enyi-2 entropy, $S_2(A) = -\log\tr\rho_A^2$, and the entropies are evaluated on the Choi state of $U$.
Since $I_3^{(2)}\geq I_3$, the butterfly effect as measured by small OTO correlators implies $I_3$-scrambling.
In \cref{sec:oto}, we show that the converse is not true: a unitary with almost maximally negative tripartite information can still have large OTO correlators.
In fact, we find that the difference $I_3^{(2)} - I_3$ can be arbitrarily large.

\begin{figure}
  \centering
  \includegraphics[width=0.25\textwidth]{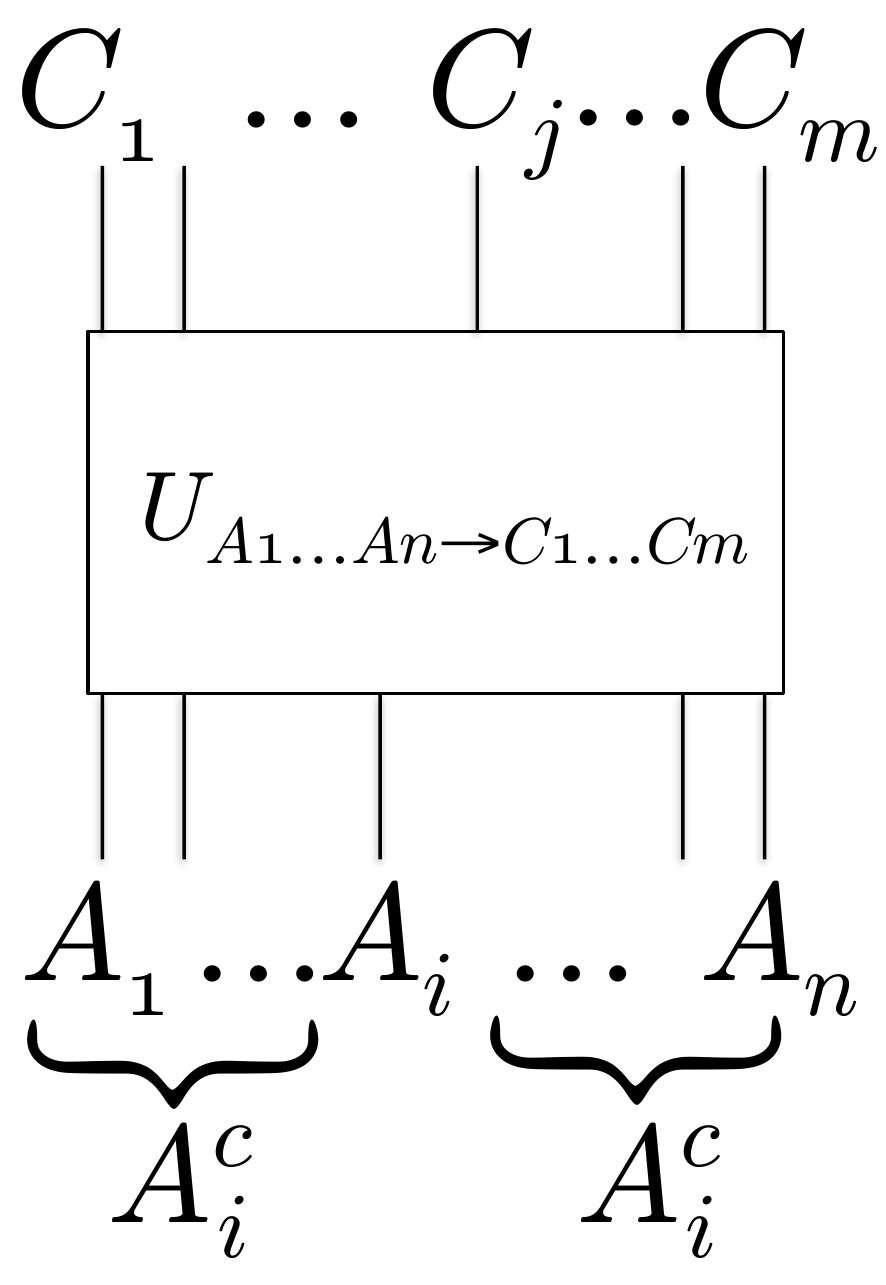}
  \caption{A multiple input and multiple output (MIMO) unitary.} 
  \label{fig:MIMO}
\end{figure}

\medskip

Finally, many of the above results can be extended to the multipartite case, as we explain in \cref{sec:multipartite}.
Let $U_{A_1\dots A_n \to C_1 \dots C_m}$ be a \emph{multiple input and multiple output (MIMO) unitary} as shown in \cref{fig:MIMO}.
We show that the natural generalization of minimal $I_3$-scrambling is to demand that $I_3(A_i;A_i^c;C_j)=0$ for all $i$ and $j$, where we write $A_i^c$ for the subset of all input subsystems save for $A_i$.
In this case, the unitary takes the following form, generalizing our result for the bipartite case:
\begin{equation}
  U_{A_1\dots{}A_n\to C_1\dots{}C_m} = \bigotimes_{i,j} U_{i\to j},
\end{equation}
where $U_{i \to j}$ is a local unitary mapping input subsystem $A_i$ to output subsystem $C_j$.
We also give an explicit construction of a family of maximally scrambling MIMO unitaries when all systems are of the same large prime dimension.

\medskip

The nonrobustness of various algebraic characterizations of chaos and scrambling, while undesirable, is one of the central messages of this article.
It typically leads to dimensional dependencies, 
which, in the context of high energy physics where Hilbert spaces are typically high-dimensional, are of particular significance.
We believe that this provides good motivation for the development of alternative, more robust characterizations and diagnostics, not only in the present context but also in the study of other quantum information concepts in high energy physics, such as quantum error correction in holographic systems.

\section{Minimal scrambling}
\label{sec:minimal}
In this section, we study properties of bipartite unitaries $U_{AB \to CD}$ where $I_3 \approx 0$.
We first consider the exact case.
%
Here, our main result is that the unitary has the following normal form: 

\begin{thm}
\label{thm:unitaries}
  A unitary $U_{AB\to CD}$ is minimally $I_3$-scrambling, i.e., $I_3=0$, if and only if it can be decomposed into a tensor product of local unitaries.
  That is,
    \[ U_{AB \to CD} = U_{A_L \to C_L} \ot U_{A_R \to D_L} \ot U_{B_L \to C_R} \ot U_{B_R \to D_R}, \]
  with respect to decompositions $A = A_L \ot A_R$, $B = B_L \ot B_R$, $C = C_L \ot C_R$, $D = D_L \ot D_R$.
  The dimensions of the subsystems are given by $\abs{A_L}=\abs{C_L}=\frac12I(A;C)_U$ etc.
\end{thm}

See \cref{fig:normal form} for an illustration.
This result is consistent with the notion of scrambling as delocalization of quantum information.
To see this, take a minimally $I_3$-scrambling unitary $U_{AB\to CD}$, and consider the residual channel
$\mathcal N_{A \to C}[\sigma_{A}] = \tr_D \Big[ U_{AB \to CD} (\sigma_{A} \ot \sigma^0_B) U_{AB \to CD}^\dagger \Big]$ 
for some choice of state $\sigma^0_B$ on $B$.
Then, \cref{thm:unitaries} implies that
\[ \mathcal N_{A \to C}[\sigma_{A}] = U_{A_L \to C_L} \sigma_{A_L} U_{A_L \to C_L}^\dagger \ot \sigma^0_{C_R}, \]
where $\sigma^0_{C_R} = U_{B_L \to C_R} \sigma^0_{B_L} U_{B_L \to C_R}^\dagger$ is independent of the channel input.
Hence, for the purposes of quantum information transfer, the residual channel $\mathcal N_{A \to C}$ is equivalent to the unitary quantum channel $U_{A_L\to C_L}$.
Likewise, $\mathcal N_{A \to D}$ is equivalent to the unitary channel $U_{A_R\to D_L}$, while $\mathcal N_{A \to CD}$ is equivalent to their tensor product.
In particular, the quantum information from $A$ can be perfectly transmitted using local decoders at $C$ and $D$, independent of the choice of input at $B$.
Thus quantum information is perfectly routed through the system in a completely localized fashion, in agreement with the absence of scrambling.

From the perspective of quantum communication, we may state this as
\[ Q_{A\to C} + Q_{A\to D} = Q_{A \to CD} = \log\,\abs{A}, \]
where $Q$ is the quantum capacity of the corresponding channels, i.e., the maximum qubit rate at which quantum communication can be transferred through the channels in the limit of many channel uses and vanishing error (see, e.g., \cite{wilde2013quantum} for details). The right-hand side equality is a consequence of unitarity.
In fact, we actually get the even stronger result that
\[ R_{A\to C} + R_{A\to D} = Q_{A \to CD} = \log\,\abs{A} \]
where $R_{A\to C}$ and $R_{A\to D}$ are simultaneously achievable, one-shot, zero-error quantum communication rates.

It is important to note that simultaneously achievable rates are different from the individual quantum capacities for general broadcast channels $A \to CD$.
The former always satisfy an inequality $R_{A \to C} + R_{A \to D} \le Q_{A \to CD}$. However, the latter need not.
This phenomenon is also found in classical communication capacities.
Consider, e.g., the basis-dependent copying channel $A\to CD$ which sends a noiseless copy of $A$ to $C$ and $D$ as $\ket{j}_A \mapsto \ket{j}_C \ket{j}_D$.
The individual capacities are $\log d$ but so is the overall capacity.
While we cannot make the same construction for quantum capacities due to the no-cloning theorem,
we can take advantage of the fact that the product of the dimensions of two subspaces can be greater than the sum to get a gap in quantum capacities as well.
Define the unitary
\begin{equation*}
  U \ket{a} \ket{b} =
  \begin{cases}
    \ket{a} \ket{b} & a,b \le d_0 \text{ or } a,b > d_0 \\
    \ket{b} \ket{a} & \text{otherwise}
  \end{cases}
\end{equation*}
where $d_0 \leq d$.
If we fix the input state $\rho_B^0 = \ket0\!\!\bra0$ then the resulting channel sends $\ket a \mapsto \ket a \ot \ket 0$ if $a \leq d_0$, and $\ket a \mapsto \ket 0 \ot \ket a$ otherwise.
Therefore, $Q_{A\to C} \geq \log d_0$ by coding in the former, $d_0$-dimensional subspace, while $Q_{A\to D}\geq\log(d-d_0)$ by coding in the latter subspace.
Hence the sum of the individual capacities is at least $Q_{A\to C}+Q_{A\to D} = \log d_0(d-d_0) > \log d$ for appropriate $d_0$.
However, $Q_{A\to CD}$ is never larger than $\log\,\abs{A} = \log d$, so we obtain the inequality $Q_{A\to C}+Q_{A\to D} > Q_{A\to CD}$.

\medskip

To prove \cref{thm:unitaries}, we first prove the corresponding statement for quantum states with vanishing conditional mutual information:

\begin{prp}
\label{prp:zero cmi}
  Any pure four-party quantum state $\rho_{ABCD}$ that satisfies the three properties
  \begin{enumerate}
    \item $I(A;B|C) = 0$, 
    \item $\rho_{AB} = \tau_{AB}$, the maximally mixed state on $AB$, and
    \item $\lvert AB \rvert = \lvert CD \rvert$.
  \end{enumerate}
  has the form
  \begin{equation*}
      \rho_{ABCD} = \Phi^+_{A_L C_L} \ot \Phi^+_{A_R D_L} \ot \Phi^+_{B_L C_R} \ot \Phi^+_{B_R D_R}
  \end{equation*}
  where $A = A_L \ot A_R$, $B = B_L \ot B_R$, $C = C_L \ot C_R$, $D = D_L \ot D_R$, and where the $\Phi^+$ denote maximally entangled states.
\end{prp}
\begin{proof}
  We note that assumptions~2 and 3 together imply that
  \begin{equation}
    \label{eq:maxmix}
    \rho_{CD} = \tau_{CD}
  \end{equation}
  and so
  \begin{equation}
    \label{eq:rank}
    \rank \rho_{ABC} = \rank \tau_D = \lvert D \rvert.
  \end{equation}
  From~\cite{hayden2004structure}, we know that if $\rho_{ABC}$ is a quantum state with $I(A;B|C)=0$ (assumption~1), then we can decompose into sectors $C = \bigoplus_i C_i$ and $C_i = C_{L_i} \ot C_{R_i}$ such that
  \begin{equation}
    \label{eq:decompo A}
    \rho_{ABC} = \sum_i p_i \, \rho^{(i)}_{A C_{L_i}} \ot \rho^{(i)}_{B C_{R_i}}.
  \end{equation}
  for some probability distribution $p_i$ and quantum states $\rho_{AC_{L_i}}^{(i)}$, $\rho_{BC_{R_i}}^{(i)}$.
  Now~\eqref{eq:rank} shows that
  \[ \lvert D \rvert = \sum_i \rank \rho^{(i)}_{A C_{L_i}} \times \rank \rho^{(i)}_{B C_{R_i}}. \]
  Thus we can decompose into sectors $D = \bigoplus_i D_i$, $D_i = D_{L_i} \ot D_{R_i}$ (where $\lvert D_{L_i} \rvert = \rank \rho^{(i)}_{A C_{L_i}}$, etc.) and purify individually to obtain a purification of $\rho_{ABC}$ of the form
  \begin{equation}
    \label{eq:puri}
    \sum_i \sqrt{p_i} \, \ket{\eta^{(i)}_{A C_{L_i} D_{L_i}}} \ot \ket{\xi^{(i)}_{B C_{R_i} D_{R_i}}}.
  \end{equation}
  By Uhlmann's theorem (see, e.g., \cite{wilde2013quantum}), the purification in~\eqref{eq:puri} only differs by a local unitary on $D$ from the four-party pure state $\rho_{ABCD}$, which likewise purifies $\rho_{ABC}$, and hence it suffices to establish the normal form for~\eqref{eq:puri}.
  Furthermore, they have the same reduced state on $CD$, namely, the maximally mixed state~\eqref{eq:maxmix}, which is unitarily invariant.
  Thus:
  \begin{equation}
    \label{eq:crucial}
    \bigoplus_{i,i'} \sqrt{p_i p_{i'}} \tr_{AB}\Big[
      \ket{\eta^{(i)}_{A C_{L_i} D_{L_i}}}\!\!\bra{\eta^{(i')}_{A C_{L_{i'}} D_{L_{i'}}}} \ot \ket{\xi^{(i)}_{B C_{R_i} D_{R_i}}}\!\!\bra{\xi^{(i')}_{B C_{R_{i'}} D_{R_{i'}}}}
    \Big] = \tau_{CD}
  \end{equation}
  We may think of the left-hand side as a big block matrix with respect to $\bigoplus_{i,j} C_i \ot D_j$ which is only supported on blocks where $i = j$.
  The right-hand side on the other hand is supported on all blocks $C_i \ot D_j$.
  Thus~\eqref{eq:crucial} can only be true if there is only a single sector (and hence no pair with $i \neq j$).
  Suppressing the index $i$, this means that, in fact, $C = C_L \ot C_R$ and $D = D_L \ot D_R$, so that \eqref{eq:decompo A} becomes
  \[ \rho_{ABC} = \rho_{AC_L} \ot \rho_{BC_R} \]
  and its purification~\eqref{eq:puri} reads
  \begin{equation}
  \label{purif1}
    \ket{\eta_{AC_LD_L}} \ot \ket{\xi_{BC_RD_R}}.
  \end{equation}
  Moreover, \eqref{eq:crucial} becomes
  \begin{equation*}
    \eta_{C_LD_L} \ot \xi_{C_RD_R} = \tau_{CD},
  \end{equation*}
  and so both $\eta_{C_LD_L} = \tau_{C_LD_L}$ and $\xi_{C_RD_R} = \tau_{C_RD_R}$ are maximally mixed.
  In particular,
  \begin{equation*}
    \lvert A \rvert \geq \lvert C_L D_L \rvert, \quad
    \lvert B \rvert \geq \lvert C_R D_R \rvert
  \end{equation*}
  by the Schmidt decomposition.
  But $\lvert AB \rvert = \lvert CD \rvert$ by assumption~3, thus in fact
  \begin{equation*}
    \lvert A \rvert = \lvert C_L D_L \rvert, \quad
    \lvert B \rvert = \lvert C_R D_R \rvert.
  \end{equation*}
  Thus $\ket{\eta_{AC_LD_L}}$ is maximally entangled between $A$ and $C_LD_L$, and $\ket{\xi_{BC_RD_R}}$ is maximally entangled between $B$ and $C_RD_R$.
  If we decompose $A = A_L \ot A_R$ and $B = B_L \ot B_R$ with $\abs{A_L} = \abs{C_L}$, $\abs{A_R}=\abs{D_L}$, etc., then we have another purification of $\eta_{C_L D_L}\ot\xi_{C_R D_R}$, given by a tensor product of maximally entangled states:
  \begin{equation}
    \big( \ket{\Phi^+_{A_L C_L}} \ot \ket{\Phi^+_{A_R D_L}} \big) \ot
    \big( \ket{\Phi^+_{B_L C_R}} \ot \ket{\Phi^+_{B_R D_R}} \big).
    \label{purif2}
  \end{equation}
  Thus, by another application of Uhlmann's theorem there exist local unitaries on $A$,$B$ that transform~\eqref{purif1} into~\eqref{purif2}.
  Absorbing all local unitaries into the tensor product decompositions, we obtain the desired result.
\end{proof}

The normal form in \cref{thm:unitaries} follows now readily from \cref{prp:zero cmi}, since the Choi state $\rho_{ABCD}$ associated with the unitary $U_{AB\to CD}$ satisfies all three assumptions of the proposition.
The formula for the dimensions of the subsystems $A_L$ etc.\ follows directly from the normal form.
For the converse, we observe that $-I_3 = I(C;D|A) = S(AC) + S(AD) - S(A) - S(B)$, where $\rho_{AC} = \Phi^+_{A_L C_L} \ot \tau_{A_R} \ot \tau_{C_R}$ and similarly for $\rho_{AD}$.
Hence, $S(AC) = \log\,\abs{A_R C_R} = \log\,\abs{A_R B_L}$ and $S(AD) = \log\,\abs{A_L D_R} = \log\,\abs{A_L B_R}$, while $S(A) = \log\,\abs A$ and $S(B) = \log\,\abs B$.
So, $S(AC) + S(AD) = \log \abs{A B} = S(A) + S(B)$, which implies that $I(C;D|A) = 0$.

\Cref{thm:unitaries} does not appear to directly generalize to isometries $V_{AB\to CD}$.
For example, consider the three-party GHZ state $\ket{\text{GHZ}}_{ACD} = (\ket{000} + \ket{111}) / \sqrt{2}$, which is the Choi state of the isometry mapping $\ket{0} \mapsto \ket{00}$ and $\ket{1} \mapsto \ket{11}$.
This is a special case of an isometry $V_{AB \to CD}$ where $B$ is trivial, and $I_3(A;C;D)$ is zero, just as for any tripartite pure state.
However, the GHZ state is clearly not of the form in \cref{prp:zero cmi}, even if we allow for maximally entangled states between $C,D$.
This can be seen by the fact that tracing out any one of the $A,C,D$ in the GHZ state gives a separable state, which is impossible for a triple of maximally entangled states unless they are all trivial.


\medskip

It is well-established in quantum information literature that the conditional information can be operationally interpreted in terms of the recoverability of quantum information for tripartite quantum states~\cite{fawzi2015quantum,sutter2015universal}. See also~\cite{winter2012stronger,kim2013application,zhang2013conditional,berta2015renyi}. In particular, it is known that, for any quantum state $\rho_{ABD}$,
\[ \norm{\rho_{ABD} - \mathcal{R}_{D \to BD} \left(\rho_{AD}\right)}_1 \le 2 \sqrt{1 - e^{-I(A;B \vert D)/2}}. \]
where $\norm{X}_1:=\tr\sqrt{X^\dagger X}$ is the \emph{trace norm} and $\mathcal R_{D\to BD}$ a quantum channel that only depends on $\rho_{BD}$~\cite{sutter2015universal}.
Applied to the Choi state of a bipartite unitary $U_{AB\to CD}$ with $-I_3\le\varepsilon$, we therefore obtain a \emph{recovery map} with
\begin{equation}
\label{eq:recovery}
  \norm{\rho_{ABD} - \mathcal{R}_{D \to BD}(\rho_{AD})}_1 \le \sqrt{2 \varepsilon}.
\end{equation}
This is immediate from \cref{thm:unitaries} when $I_3 =0$.
This recovery property of the state from local information is in stark contrast with the maximally scrambling case, such as in the model of black hole evaporation from~\cite{hayden2007black}, and we discuss this in more detail on p.~\pageref{black holes}. 

In contrast to the interpretation in terms of recovery maps, \cref{thm:unitaries} itself is not robust in the sense that there exist unitaries for which $I_3$ is arbitrarily close to zero, while their distance to any unitary of the form of \cref{thm:unitaries} stays bounded away from zero.
Here, we measure distance using the \emph{diamond norm} between two quantum channels $\mathcal N$ and $\mathcal M$,
\begin{equation}
\label{eq:diamond}
  \norm{\mathcal N - \mathcal M}_\diamond = \max_n \max_{\rho_{AR}} \norm{\bigl( \id_{R} \ot \mathcal N_{A \to B}  - \id_{R} \ot \mathcal M_{A \to B} \bigr)(\rho)}_1,
\end{equation}
where we optimize over all states $\rho$ on $AR$, with $R$ an auxiliary $n$-dimensional Hilbert space ($n = \abs A$ is sufficient).
As the trace distance quantifies how well one can experimentally distinguish quantum states \cite{helstrom1969quantum}, the diamond norm is a natural measure of how well one can distinguish two quantum channels even with an auxiliary system.

Our construction is explicit and goes as follows.
We choose $A=B=C=D=\C^d$ and define a bipartite unitary $U_d$ that is maximally $I_3$-scrambling on some subspace and the identity otherwise.
More precisely,
\begin{equation}
  \label{eq:counterU}
  U_d \ket a \ket b = \begin{cases}
    U_S \ket{a} \ket{b} & 0 \leq a, b < d_S\\
    \ket{a} \ket{b} & \text{otherwise}
  \end{cases}
\end{equation}
for some $d_S\leq d$, where $U_S$ is a bipartite unitary $A_S B_S \to C_S D_S$ that is maximally $I_3$-scrambling, i.e., $I_3 = -2\log d_S$, with $A_S$ the subspace spanned by the first $d_S$ basis vectors of $A$, etc.
We prove the existence of such unitaries for arbitrary odd dimension $d_S$ in \cref{sec:maximal} below.
Then we have the following result:

\begin{prp}
  \label{prp:nonRobustZero}
  Let $d_S$ be an odd constant. Then the bipartite unitaries $U_d$ defined in~\eqref{eq:counterU} satisfy
  \begin{equation*}
    \lim_{d \to \infty} I_3(A;B;C)_{U_d} = 0.
   \end{equation*}
  However,
  \begin{equation*}
    \liminf_{d \to \infty} \inf_{U_0} \Vert U_d - U_0 \Vert_\diamond \geq 1> 0,
   \end{equation*}
  where the infimum is over all unitaries $U_0$ with vanishing tripartite information.
\end{prp}

That is, by making $U_d$ $I_3$-scrambling on a subspace whose relative size goes to zero for large $d$, we can make the triparite information go to zero while still leaving a nonzero subspace that is $I_3$-scrambling, thereby keeping the diamond norm finitely bounded from zero. It is also interesting to note that the Choi state of $U_d$ converges to that of the identity channel, a quantum Markov chain state, in trace distance, while the channel itself does not converge to the identity nor any minimally $I_3$-scrambling unitary in diamond norm.

On the other hand, we note that in terms of simultaneous local one-shot quantum capacities of $U_d$, $\lim_{d \to \infty} Q_{A \to CD} - (R_{A \to C} + R_{A \to D}) =0$.
Indeed, by coding in the complementary subspace of $A_S$, $R_{A \to C} \ge \log(d-d_S)$ can be achieved.
Asymptotically, this goes like $\log d$, since
\begin{equation*}
  \lim_{d\to\infty} \log d - \log \bigl(d-d_S\bigr) = - \lim_{d\to\infty} \log \Bigl( 1 - \frac {d_S}d \Bigr) = 0.
\end{equation*}
Thus, since $R_{A \to D} \ge 0$, $\lim_{d\to\infty} Q_{A \to CD} - (R_{A \to C} + R_{A \to D}) \le 0$.
The other inequality is trivial, so we have equality.
Hence, one might be tempted to interpret $I_3$ as the difference between the sum of the simultaneous local quantum capacities $A \to C, D$ and the maximum possible value $\log \abs{A}$, which is true in this example for the limit of large $d$.
For finite $d$, however, we can find examples where this interpretation fails.

The interpretation can be partially salvaged, however, by considering instead entanglement-assisted classical communication with random codes generated using maximally entangled states while fixing the input to $B$ to be maximally mixed.
This follows from the observation
\begin{equation}
\label{eq:entang-assist}
  I_3 = I(A;C) + I(A;D) - I(A;CD) = I(A;C) + I(A;D) - 2\log\abs{A}
\end{equation}
and the fact that the entanglement-assisted classical communication rate of a channel $\mathcal{N}_{A\to C}$ using such a code is given by the mutual information $I(A;C)$ of its Choi state~\cite{bennett1999entanglement, bennett2002entanglement}.
Since the mutual information $I(A;CD) = 2 \log\abs{A}$ is as large as it can be, it is not just an achievable rate but in fact the capacity of the $A \rightarrow CD$ channel.
\Cref{eq:entang-assist} therefore states that the sum of the two entanglement-assisted achievable rates is bounded above by the entanglement-assisted capacity.

\Cref{prp:nonRobustZero} is a consequence of the following technical estimates proved in \cref{app:minimal}:

\newcommand{\lemmaNonRobustZero}[4]{
\begin{#1}
  #2
  Consider the unitaries $U_d$ from~\eqref{eq:counterU} and their Choi states $\rho_{ABCD,d}$. Then,
  \begin{equation}
  #3
    \norm{\rho_{ABCD,d} - \Phi^+_{AC} \ot \Phi^+_{BD}}_1 \le 4 \, \frac{d_S}{d}
   \end{equation}
  and
  \begin{equation}
  #4
    \inf_{U_0} \Vert U_d - U_0 \Vert_\diamond \geq 1 - \frac {2 + 2\log d_S} {\log d}
   \end{equation}
  where the infimum is over all unitaries $U_0$ with vanishing tripartite information.
\end{#1}
}
\newtheorem*{nonRobustZero}{\Cref{lem:nonRobustZero}}
\lemmaNonRobustZero{lem}{\label{lem:nonRobustZero}}{\label{eq:nonRobustZero state dist}}{\label{eq:nonRobustZero diamond dist}}

Indeed, \eqref{eq:nonRobustZero state dist} implies that the difference between the subsystem entropies vanishes in the limit of large $d$.
This follows from the Fannes-Audenaert inequality~\cite{fannes1973continuity,audenaert2007sharp}, which asserts that, for any two quantum states $\rho$ and $\sigma$ on a $D$-dimensional Hilbert space,
\begin{equation}
\label{eq:fannes audenaert}
  \abs{S(\rho) - S(\sigma)} \leq T \log(D-1) + h(T),
\end{equation}
where $T = \frac12\norm{\rho-\sigma}_1$ and $h(T)=-T\log T-(1-T)\log(1-T)$ is the binary entropy function, which can be upper bounded as $h(T) \leq 2\sqrt T$.
But $\Phi^+_{AC} \ot \Phi^+_{BD}$ is the Choi state of the identity channel, which has zero tripartite information.
Hence the tripartite information $I(A;B;C)_{U_d}$ goes to zero in the limit of large $d$.
In the same limit, the right-hand side of \eqref{eq:nonRobustZero diamond dist} converges to $1$.
This concludes the proof of \cref{prp:nonRobustZero}.

\section{Maximal scrambling}
\label{sec:maximal}
We now consider the opposite extreme where $I_3 \approx -2 \log \min\{\abs{A},\dots,\abs{D}\}$ and compare it to the results we obtained in the minimally $I_3$-scrambling case.
Note that this is the most negative value it can take since
\begin{equation}
\label{eq:max i3 lower bound}
\begin{aligned}
  I_3
  &= -I(A;B|C)
  = S(C) + S(ABC) - S(AC) - S(BC) \\
  &= S(C) + S(D) - S(AC) - S(AD) \\
  &= I(A;C) + I(A;D) - 2S(A) \\
  &= I(A;C) + I(A;D) - 2\log\abs A \\
  &\geq -2\log \abs{A}
\end{aligned}
\end{equation}
since the mutual information is always nonnegative.
A similar inequality holds for the other subsystems.

We first discuss the existence of maximally $I_3$-scrambling unitaries in the case where $A=B=C=D=\CC^d$.
Clearly, $I_3 = -2 \log d$ if and only if any bipartite subsystem is maximally mixed, i.e., if $S(AB)=S(AC)=\dots=2\log d$.
Such unitaries are precisely four-party perfect tensors, i.e., tensors that are unitary from any bipartition to the complement, as pointed out in~\cite{hosur2015chaos}.
This establishes the existence of maximally $I_3$-scrambling unitaries in sufficiently large prime dimension $d$, since a stabilizer state chosen at random will be a perfect tensor with high probability~\cite{hayden2016holographic}.
On the other hand, the following explicit construction achieves the same for any odd dimension $d$:
\begin{equation}
\label{eq:odd scrambling}
  U_S \ket{i}_A \ket{j}_B = \ket{i+j}_C \ket{i-j}_D,
\end{equation}
where all arithmetic is modulo $d$.
We require $d$ to be odd so that $U_S$ is unitary.
It can be readily verified that $I_3 = -2\log d$.
We note that~\eqref{eq:odd scrambling} is a straightforward generalization of the three-qutrit code from~\cite{keet2010quantum}.
It is interesting to observe that $U_S^2$ is minimally $I_3$-scrambling.
In this sense, a unitary that is maximally $I_3$-scrambling can still have a very small recurrence time.

The relationship to quantum error correcting codes can also be used to argue that there exists no maximally $I_3$-scrambling unitary for qubits ($d=2$).
Indeed, assume that such a unitary $U_{AB\to CD}$ exists and consider the isometry $V_{A\to BCD} := U_{AB'\to CD} \ket{\Phi^+_{BB'}}$ obtained by inputting one half of a maximally entangled state into $B$.
Then the perfect tensor property implies that we can correct for the erasure of any one of the output qubits $B$, $C$ and $D$.
In other words, $V_{A\to BCD}$ would be a code for the qubit erasure channel of length 3.
But this is ruled out by~\cite{grassl1997codes}.
Hence, such a $U$ does not exist.


\medskip

We return to the general setup, where the dimensions of the systems $A,\dots,D$ need not be equal, and consider the consequences of a unitary being maximally $I_3$-scrambling.
In particular, we consider the residual channels from a single input to a single output.
Then, we expect the channels residual channels $A \to C$ etc.\ to be noisy since quantum information should be delocalized.
Indeed, we find:

\begin{prp}
\label{prp:depolarizing}
  Let $U_{AB\to CD}$ be a maximally $I_3$-scrambling unitary and $\rho_{ABCD}$ its Choi state.
  If either $A$ or $C$ have the smallest dimension among the four subsystems then $\rho_{AC}$ is maximally mixed and $I(A;C)=0$.

  As a consequence, the residual channel $\mathcal N_{A\to C}[\sigma_A]=\tr_D[U (\sigma_A \ot \tau_B) U^\dagger]$ corresponding to the maximally mixed input on $B$ is completely depolarizing, i.e., its channel output is the maximally mixed state $\tau_C$ for any input state $\sigma_A$.%
\footnote{Dually, $\mathcal N_{B\to C}[\sigma_B]=\tr_D[U(\sigma^0_A\ot\sigma_B)U^\dagger]$ maps $\tau_B \mapsto \tau_C$ for \emph{any} choice of input state $\sigma^0_A$ at $A$.}
\end{prp}
\begin{proof}
  If the dimension of $A$ is smallest, maximal $I_3$-scrambling means that $I_3 = -2\log\abs A$.
  Thus it follows from~\eqref{eq:max i3 lower bound} that $I(A;C)=I(A;D)=0$, since the mutual information is always nonnegative.
  Similarly, if $C$ is smallest then we have $I_3=-2\log\abs C$, which implies that $I(A;C) = I(B;C) = 0$.

  In either case, we thus find that $I(A;C)=0$ and hence that $\rho_{AC}=\rho_A\ot\rho_C=\tau_{AC}$, since both $\rho_A$ and $\rho_C$ are maximally mixed.
  To see that this implies the second claim, we note that $\rho_{AC}$ is the Choi state of the residual channel $\mathcal N_{A\to C}$.
  Hence, $\mathcal N_{A'\to C}[\Phi^+_{AA'}] = \tau_{AC}$ and therefore $\mathcal N_{A\to C}[\sigma_A] = \tau_C$ for any input state $\sigma_A$.
\end{proof}


Completely depolarizing channels have zero capacity of any kind, in agreement with our expectation that the quantum information at $A$ gets fully delocalized for maximally mixed input at $B$.
In \cref{app:typical} we show that if $\abs{D} \gg \abs{AC}$ then $\rho_{AC}\approx\tau_{AC}$ for typical input states on $B$.
Moreover, if $\abs{D} \gg \abs{AC}^2$ then the residual channel $\mathcal N_{A\to C}$ is typically entanglement-breaking, in which case it still has zero quantum capacity.

In general, there exist input states on $B$ such that the corresponding residual channel $A\to C$ can still be used for communication.
For example, consider the unitary defined in~\eqref{eq:odd scrambling}.
If we fix the input on $B$ to a computational basis state $\ket0$, then
\begin{align*}
  \mathcal N_{A\to C}[\rho_A]
= \tr_D\big[U_S (\rho_A \ot \ket0\!\!\bra0) U_S^\dagger\big]
= \sum_i \braket{i|\rho_A|i} \ket i\!\!\bra i_C.
\end{align*}
Hence, the residual channel is the completely dephasing channel, which has maximal classical capacity. If we instead fix the input on $B$ to be in the state $\frac{1}{\sqrt{3}}(\ket{0} + \sqrt{2} \ket{1})$ and consider the $d=3$ case, we obtain a residual channel $A \to C$ with positive quantum capacity. To see this, we use the fact that the coherent information of a channel is a lower bound on the quantum capacity~\cite{lloyd1997capacity,shor2002quantum,devetak2005private}:
\begin{equation}
  Q(\mathcal{N}_{A\to C}) \ge I(\mathcal{N}_{A \to C}) \equiv \max_{\varphi_{RA}} I( R \rangle C)_{\mathcal{N}(\varphi)}
  \label{}
\end{equation}
where $I(R \rangle C) \equiv S(C) - S(RC)$ is the \textit{coherent information}. If we choose the input state $\ket{\varphi}_{RA} = \frac{1}{\sqrt{3}} (\ket{00} + \sqrt{2}\ket{11})$, we obtain $I(R \rangle C) = \frac{12}{9} - \frac{5}{9} \log 5 > 0$.

\medskip

We can also interpret \cref{prp:depolarizing} from the perspective of recovery of quantum information.
If we assume that the dimension of $A$ is smallest then both residual channels $A\to C$ and $A\to D$ are completely depolarizing.
Given only $D$, none of the quantum information at $A$ can be recovered, while if we supplement it with $B$, perfect recovery is possible.
More precisely, we can transfer entanglement from $A$ to $BD$ perfectly.
This follows from the fact that $\rho_{ABCD}$ and $\Phi^+_{AA'} \ot \Phi^+_{CC'}$ both purify the reduced state $\rho_{AC} = \tau_A\ot\tau_C$, which by Uhlmann's theorem implies the existence of a decoding operation $\mathcal D_{BD\to A'}$.

\phantomsection\label{black holes}
One of the motivations for studying scrambling unitaries comes from black hole physics.
The preceding interpretation applies naturally to the model of black hole evaporation in~\cite{hayden2007black} and was also discussed in~\cite{hosur2015chaos}.
We can schematically model black hole evaporation by a bipartite unitary time evolution where $A$ is half of a Bell pair whose other half $A'$ enters the black hole at time $t_0$,
$B$ is the Hawking radiation emitted before $t_0$, assumed to be maximally entangled with the black hole $B$ at $t_0$,
$C$ is the state of the remaining black hole at a later time $t_1$,
and $D$ is the Hawking radiation emitted in the interval $[t_0, t_1]$.
All indications are that black holes are highly scrambling~\cite{hayden2007black,sekino2008fast,shenker2014black,roberts2015diagnosing}. 
If we assume that they are maximally $I_3$-scrambling then we find that $A'$ cannot be recovered from the late-time Hawking radiation $D$ alone, while it would be possible when also given the old Hawking radiation $B$.
In contrast, if the process were minimally $I_3$-scrambling then someone without knowledge of quantum state at $A$ and with only the new Hawking radiation $D$ could apply a local operation $\mathcal R_{D \to B D}$ to approximately recover the old Hawking radiation, so that the overall tripartite state $\mathcal{R}_{D \to BD} (\rho_{AD})$ is close to $\rho_{ABD}$ (\cref{eq:recovery}).

\medskip

Lastly, we consider the approximate case, where $I_3 \approx -2 \log \min\{\abs{A},\dots,\abs{D}\}$.
For concreteness, we assume that the dimension of system $A$ is smallest among all four subsystems and $I_3 = -2\log\abs A+\varepsilon$.
Then, $I(A;D)\leq\varepsilon$ as a consequence of~\eqref{eq:max i3 lower bound} (cf.\ the proof of \cref{prp:depolarizing}).
Using Pinsker's inequality, this implies that $\norm{\rho_{AD} - \tau_A\ot\tau_D}_1 \leq \sqrt{2 \ln(2)\varepsilon}$.
In particular, if we put one half of a maximally entangled state into the residual channel $A\to D$, then the resulting state is close to being completely uncorrelated.
Likewise, $\rho_{AC} \approx \tau_A\ot\tau_C$, and hence $\rho_{ABCD}$ and $\Phi^+_{AA'} \ot \Phi^+_{CC'}$ still purify approximately the same state.
It follows, again by Uhlmann's theorem, that there still exists a quantum operation $\mathcal D_{BC\to A'}$ such that $\mathcal D_{BD\to A'}[\rho_{ABD}] \approx \Phi^+_{AA'}$.
In this sense, the recovery interpretation described above can be made robust.

On the other hand, the stronger conclusion of \cref{prp:depolarizing} is not robust in the sense that we can find unitaries such that the negative tripartite information goes to its maximal value, while the diamond norm~\eqref{eq:diamond} between the residual channel $\mathcal N_{A\to C}$ and the completely depolarizing channel remains finite. Furthermore, we find that there are such unitaries with nonvanishing one-shot zero-error quantum capacity.
That is, a unitary can be arbitrarily close to being maximally $I_3$-scrambling even though its residual channel can still transmit quantum information perfectly at a nonvanishing rate.
The sequence of unitaries we use is again \eqref{eq:counterU},
\begin{equation*}
  U_d \ket a \ket b = \begin{cases}
    U_S \ket{a} \ket{b} & 0 \leq a, b < d_S\\
    \ket{a} \ket{b} & \text{otherwise},
  \end{cases},
\end{equation*}
except this time $d_S$ will be large.
We still require that $d_S = d-d_0$ is odd, so that the existence of a maximally $I_3$-scrambling unitary $U_S$ is guaranteed.
Then we can then establish the following result:

\begin{prp}
  \label{prp:nonRobustMax}
  Let $d_0$ be a constant and consider the family of unitaries $U_d$ for odd $d_S = d-d_0$.
  Then,
  \begin{equation}
  \label{eq:nonRobustMax I_3}
    \lim_{d \to \infty} \bigl( I_3(A;B;C)_{U_d} + 2\log d \bigr) = 0,
   \end{equation}
  while the residual channels $\mathcal N_{A\to C,d}[\sigma_A]=\tr_D[U_d (\sigma_A \ot \tau_B) U_d^\dagger]$ have bounded distance away from the completely depolarizing channel $\Delta_{A\to C}$:
  \begin{equation}
  \label{eq:nonRobustMax diamond}
    \lim_{d \to \infty} \norm{\mathcal{N}_{A \to C} - \Delta_{A \to C}}_\diamond = 2 > 0.
  \end{equation}
  Moreover, their one-shot zero-error quantum capacities $Q_{A\to C,d}$ can be lower bounded as
  \begin{equation}
  \label{eq:nonRobustMax R}
    Q_{A\to C,d} \geq \log d_0 > 0.
   \end{equation}
\end{prp}

To establish \cref{prp:nonRobustMax}, we first note that the last bound~\eqref{eq:nonRobustMax R} is immediate, since we can code perfectly using the $d_0$-dimensional subspaces.
The first two bounds, \eqref{eq:nonRobustMax I_3} and \eqref{eq:nonRobustMax diamond}, follow from the following lemma, proved in \cref{app:maximal}, together with the Fannes-Audenaert inequality~\eqref{eq:fannes audenaert} that we similarly used to establish \cref{prp:nonRobustZero}.

\newcommand{\lemmaNonRobustMax}[2]{
\begin{#1}
  #2
  Consider the unitaries $U_d$ from~\eqref{eq:counterU} and their Choi states $\rho_{ABCD,d}$.
  Then $\rho_{AD,d}$ is maximally mixed, and
  \begin{equation*}
    \norm{\rho_{AC,d} - \tau_{AC}}_1 \leq 8 \frac {d_0} d,
  \end{equation*}
  where $d_0 = d - d_S$.
  On the other hand, if $d_S < d$ then
  \begin{equation*}
    \norm{\mathcal{N}_{A \to C} - \Delta_{A \to C}}_\diamond \geq 2 - \frac{2}{d}.
   \end{equation*}
\end{#1}
}

\newtheorem*{nonRobustMax}{\Cref{lem:nonRobustMax}}
\lemmaNonRobustMax{lem}{\label{lem:nonRobustMax}}

\section{Tripartite information and state redistribution}
\label{sec:general}

We now briefly discuss the meaning of general values of the tripartite information.
Naturally, we would like to look for operational interpretations that hold in general.
Using the equivalence between tripartite information and conditional mutual information, \eqref{eq:neg tri is cmi}, one such interpretation is given by the task of \emph{quantum state redistribution}, in which a party holding two quantum systems is to transfer one of the systems to a party holding one~\cite{devetak2008exact}.
Specifically, given many copies of a quantum state $\rho_{ACD}$ with purification $\rho_{ABCD}$, a party with $AC$ can transmit $A$ to a party with $D$ using a rate of $\frac{1}{2}I(A;B|D)$ qubits of communication, $\frac{1}{2} I(A;C) - \frac{1}{2}I(A;D)$ ebits (i.e., shared Bell pairs of maximally entangled qubits) and no classical communication.
Conversely, $\frac{1}{2} I(A;B|D)$ is the minimum rate of quantum communication required by any state redistribution protocol.
This is consistent with the intuition of scrambling --- a strongly scrambling unitary will delocalize the information from the inputs so that observers at individual outputs have little knowledge of the inputs.
Hence, a large number of qubits should be required to transmit this information.

We can cross-check this intuition with our main results in the minimally and maximally scrambling cases and give explicit protocols in each case.
For the minimally $I_3$-scrambling case, we cross-check \cref{thm:unitaries} by applying this result to the reduced Choi state $\rho_{ACD}$ of the unitary.
Using the above result, to transfer $A$ from $AC$ to $D$, we shouldn't need any communication and consume $\log \frac{\abs{A_L}}{\abs{A_R}}$ ebits, where we are using the notation of \cref{thm:unitaries}. This is consistent with our result as we can prepare $\ket{\Phi_{A_R D_L}}$ locally.
Thus, we only need to consume $\log\abs{A_L}$ ebits to transmit $A_L$.
However, we can use the $\log \abs{A_R}$ pre-existing ebits to transmit $A_L$ for a net ebit cost of $\log \frac{\abs{A_L}}{\abs{A_R}}$.
No communication was done, so our qubit and bit costs are indeed zero.

In the maximally $I_3$-scrambling case, we can cross-check with \cref{prp:depolarizing}.
In the case where $A$ is the smallest system, \cite{devetak2008exact} states that we should need $\log \abs{A}$ qubits, zero ebits, and zero bits.
This is achieved by the trivial protocol that transfers $A$ to $D$ over a quantum channel, in agreement with our result.

\section{Tripartite information and OTO correlators}
\label{sec:oto}

An important property of the definition of scrambling using the tripartite information is that it can be related to scrambling as measured by out-of-time-order (OTO) correlators, as explained in the introduction.
Specifically, we recall the following formula for the product of average OTO correlators,
\[
\lvert\braket{\mathcal O_C(t) \mathcal O_A \mathcal O_C(t) \mathcal O_A}_{\beta=0}\rvert
\times
\lvert\braket{\mathcal O_D(t) \mathcal O_A \mathcal O_D(t) \mathcal O_A}_{\beta=0}\rvert
\propto 2^{I_3^{(2)}},
\]
where
\begin{equation}
\label{eq:renyi3 defn}
  I_3^{(2)} = S_2(A) + S_2(B) - S_2(AC) - S_2(AD) = \log\abs A + \log\abs B - S_2(AC) - S_2(AD)
\end{equation}
is a R\'enyi-2 version of the tripartite information, defined in terms of the \emph{R\'enyi-2 entropy} $S_2(\rho) = -\log\tr\rho^2$ instead of the von Neumann entropy.
Since $S_2(\rho) \leq S(\rho)$ for any quantum state $\rho$, one obtains that $I_3^{(2)}\geq I_3$.
Thus the `butterfly effect' as measured by small OTO correlators implies $I_3$-scrambling~\cite{hosur2015chaos}.

However, the converse of this statement is not true.
That is, a $I_3$-scrambling bipartite unitary can nevertheless have high OTO correlators.
One example of this is again given by the family of unitaries $U_d$ defined in~\eqref{eq:counterU}, where we find an arbitrarily large gap between $I_3$ and $I_3^{(2)}$.

\begin{prp}
  \label{prp:renyiGap}
  Consider the unitaries $U_d$ defined in~\eqref{eq:counterU} and choose $d_0 \sim \sqrt[4]{d}$.
  Then
  \[ I_3^{(2)}(A;B;C)_{U_d} - I_3(A;B;C)_{U_d} \gtrsim \frac12\log d, \]
  in the limit of large $d$.
\end{prp}

This is proved by explicit calculation in \cref{app:renyi}, where we find that for sufficiently large $d$,
\begin{equation}
\label{eq:renyi tripartite}
  I_3^{(2)}(A;B;C)_{U_d} \geq -\frac32 \log d.
\end{equation}
On the other hand, $I_3(A;B;C)_{U_d} \sim -2 \log d$ as a consequence of \cref{eq:nonRobustZero state dist,eq:fannes audenaert}.
Together this establishes \cref{prp:renyiGap}.

This large separation can be understood by the fact that we have large individual OTO correlators.
To see this, it is useful to choose bases of local Hermitian operators, $\tr \mathcal O_{D,i} \mathcal O_{D,j} = d \delta_{i,j}$ etc., that are adapted to the scrambling and nonscrambling subspaces.
Indeed, we can write
\[ U_d = U_S \op I_{\bar S}, \]
where $U_S$ is the maximally $I_3$-scrambling unitary acting on $A_SB_S = C_SD_S$ and $I_{\bar S}$ the identity operator on the complement $\overline{C_SD_S} = C_SD_0 \op C_0D_S \op C_0D_0$.
Hence, if $\mathcal O_{D,i}$ is an operator that only acts on $D_0$, it will commute with $U_d$, so that $\mathcal O_{D,i}(t) = \mathcal O_{D,i}$.
In this case, it follows that, for any local operator $\mathcal O_A$ on $A$,
\begin{equation*}
  \braket{\mathcal O_{D,i}(t) \mathcal O_A \mathcal O_{D,i}(t) \mathcal O_A}_{\beta=0}
= \braket{\mathcal O_{D,i} \mathcal O_A \mathcal O_{D,i} \mathcal O_A}_{\beta=0}
= \frac 1 {d} \tr \mathcal O_{D,i}^2 \times \frac 1 {d} \tr \mathcal O_{A,i}^2
= 1.
\end{equation*}
Furthermore, the number of such pairs of maximally correlated operators will be increasing without bound as $d \to \infty$.

\section{Multipartite generalizations}
\label{sec:multipartite}

The main results for the minimal and maximal cases above can be generalized to the multipartite setting.
However, it is not clear, a priori, how to extend the definition of $I_3$-scrambling to the MIMO case.
In the following, we will justify defining $I_3$-scrambling for \emph{multiple input and multiple output (MIMO) unitaries} $U_{A_1\dots A_n \to C_1 \dots C_m}$ using tripartite informations of the form
\begin{equation*}
  -I_3(A_i ; A_i^c ; C_j) = I(A_i ; A_i^c | C_j) = I(C_j ; C_j^c | A_i),
\end{equation*}
where $A_i^c$ is the subset of all input subsystems save for $A_i$ and $C_j^c$ the subset of all output subsystems except for $C_j$ (\cref{fig:MIMO}).
The equalities follow from the bipartite case, \eqref{eq:neg tri is cmi}, if we partition the Choi state of $U$ into the four subsystems $A_i, A_i^c, C_j, C_j^c$.

\subsection*{Minimal scrambling}
We define a \emph{minimally $I_3$-scrambling} MIMO unitary to be a unitary $U_{A_1 \dots A_n \to C_1 \dots C_m}$ such that
\[ I_3(A_i ; A_i^c ; C_j) = 0 \]
for all $i,j$.
Again, we find that such a unitary can be decomposed into a tensor product of local unitaries connecting individual inputs and outputs, generalizing \cref{thm:unitaries}:

\begin{thm}
  \label{thm:mimoUnitaries}
  Let $U_{A_1\dots{}A_n\to C_1\dots{}C_m}$ be a MIMO unitary.
  Then $U$ is minimally $I_3$-scrambling if and only if it is of the form
  \[ U_{A_1\dots{}A_n\to C_1\dots{}C_m} = \bigotimes_{i,j} U_{i\to j} \]
  with respect to decompositions
  $A_i = \bigotimes_{j=1}^m A_{i\to j}$ for $i=1,\dots,n$,
  $C_j = \bigotimes_{i=1}^n C_{i\to j}$ for $j=1,\dots,m$
  and unitaries
  $U_{i\to j} \colon A_{i\to j} \to C_{i\to j}$ for $i,j$.
\end{thm}

We will prove \cref{thm:mimoUnitaries} by viewing the MIMO unitary as a bipartite unitary where we group inputs and outputs.
This will then allow us to iteratively apply \cref{thm:unitaries} to decompose the MIMO unitary piece by piece.
We will first peel off all the unitaries for a single input and then repeat for all other inputs.
To do so, we need to show that we can decompose a MIMO unitary into a local unitary and a residual MIMO unitary such that $A_1$ and $C_1$ have zero mutual information on the residual unitary and such that the residual MIMO unitary is still minimally $I_3$-scrambling:

\begin{lem}
\label{lem:multiparty one to one}
  Let $U_{A_1\dots{}A_n\to C_1\dots{}C_m}$ be a minimally $I_3$-scrambling MIMO unitary.
  Then there exist decompositions
  $A_1 = A_{1\to 1} \ot A'_1$,
  $C_1 = C_{1\to 1} \ot C'_1$
  and unitaries
  $U_{1\to 1} \colon A_{1\to 1} \to C_{1\to 1}$,
  $U'_{A'_1A_2\dots{}A_n\to C'_1C_2\dots{}C_m}$
  such that
  \[ U_{A_1\dots{}A_n\to C_1\dots{}C_m} = U_{1\to 1} \ot U'_{A'_1A_2\dots{}A_n\to C'_1C_2\dots{}C_m}. \]
  Here, $U'$ is a minimally $I_3$-scrambling MIMO unitary that satisfies $I(A'_1;C'_1)_{U'} = 0$.
\end{lem}
\begin{proof}
  We apply \cref{thm:unitaries} with $A = A_1$, $B=A_2\dots A_n$ and $C=C_1$ and $D=C_2\dots C_m$.
  Thus we obtain that
  \begin{equation}
  \label{eq:nofo}
    U_{AB\to CD} = U_{A_L \to C_L} \ot \bigl( U_{A_R \to D_L} \ot U_{B_L \to C_R} \ot U_{B_R\to D_R} \bigr).
  \end{equation}
  If we define $A_{1\to 1} := A_L$, $A'_1 := A_R$, $C_{1\to 1} := C_L$, $C'_1 := C_R$, $U_{1\to 1} := U_{A_L\to C_L}$ and $U'$ as the tensor product of the three unitaries on the right-hand side then we obtain a decomposition as in the statement of the lemma.

  That $U'$ is still minimally $I_3$-scrambling follows from the fact local unitaries $U_{1\to 1}$ and the overall unitary $U$ have zero tripartite information, in addition to the additivity of von Neumann entropy for tensor product states  (cf.~\cite{bao2015holographic,nezami2016multipartite}).
  And the statement about the mutual information holds because $I(A_R ; C_R)_{U'}=0$ by direct inspection of the normal form~\eqref{eq:nofo}.
\end{proof}

By iteratively applying \cref{lem:multiparty one to one}, we find decompositions $A_1=\bigotimes_{j=1}^m A_{1\to j} \ot A'_1$ and $C_j = C_{1\to j} \ot C'_j$ such that $U$ factors into a tensor product
\[ U_{A_1\dots{}A_n\to C_1\dots{}C_m} = \bigotimes_{j=1}^m U_{1\to j} \ot U'_{A'_1A_2\dots{}A_n\to C'_1\dots{}C'_m} \]
of local unitaries $U_{1\to j} \colon A_{1\to j} \to C_{1\to j}$ with a residual unitary $U'$.
The latter is minimally $I_3$-scrambling and moreover satisfies $I(A'_1;C'_j)_{U'}=0$ for all $j$ (using monotonicity of the mutual information).
However, we also need to make sure that this process will consume all of $A_1$.
This is a consequence of the following lemma, applied to the residual unitary $U'$.

\begin{lem}
  Let $U_{A_1\dots{}A_n\to C_1\dots{}C_m}$ be a minimally $I_3$-scrambling MIMO unitary with $I(A_1;C_j)=0$ for $j=1,\dots,m$.
  Then the system $A_1$ is trivial.
\end{lem}
\begin{proof}
  First note that, for all $j=1,\dots,m$,
  \begin{equation}
  \label{eq:important}
  \begin{aligned}
    0 &= I_3(A_1;A_1^c;C_j) = I(A_1;A_1^c|C_j) \\
    &= S(A_1C_j) + S(A_1^cC_j) - S(C_j) - S(A_1A_1^cC_j) \\
    &= S(A_1C_j) + S(A_1C_j^c) - S(C_j) - S(C_j^c) \\
    &= S(A_1) + S(A_1C_j^c) - S(C_j^c),
  \end{aligned}
  \end{equation}
  where the last equality follows from the assumption that $I(A_1;C_j)=0$.
  This implies the following recursion formula:
  \begin{align*}
    &\quad S(A_1C_j\dots{}C_m) - S(C_j\dots{}C_m) \\
    &= S(A_1C_j\dots{}C_m) + S(A_1C_j^c) - S(C_j\dots{}C_m) - S(C_j^c) + S(A_1) \\
    &\geq S(A_1C_1\dots{}C_m) + S(A_1C_{j+1}\dots{}C_m) - S(C_j\dots{}C_m) - S(C_j^c) + S(A_1) \\
    &= S(A_1C_{j+1}\dots{}C_m) - S(C_{j+1}\dots{}C_m) - S(C_j) - S(C_j^c) + S(A_1) + S(A_1^c) \\
    &= S(A_1C_{j+1}\dots{}C_m) - S(C_{j+1}\dots{}C_m)
  \end{align*}
  The first equality holds by plugging in~\eqref{eq:important}, the inequality is strong subadditivity, and the last two follow by using that the reduced state $\rho_{C_1\dots{}C_m}$ is maximally mixed by unitarity.
  If we start with~\eqref{eq:important} for $j=1$ and successively apply the recursion formula, we obtain
  \begin{align*}
  0 &= S(A_1) + S(A_1C_2\dots{}C_m) - S(C_2\dots{}C_m) \\
  &\geq S(A_1) + S(A_1C_3\dots{}C_m) - S(C_3\dots{}C_m) \\
  &\geq \dots \geq 2 S(A_1).
  \end{align*}
  We conclude that $\log\abs{A_1} = S(A_1) = 0$.
\end{proof}

The above considerations thus allow us to completely peel off $A_1$ from the MIMO unitary, leaving a minimally $I_3$-scrambling MIMO unitary on the other inputs.
We have thus proved the following lemma:

\begin{lem}
\label{lem:multipartite one to many}
  Let $U_{A_1\dots{}A_n\to C_1\dots{}C_m}$ be a minimally $I_3$-scrambling MIMO unitary.
  Then there exist decompositions
  $A_1 = \bigotimes_{j=1}^m A_{1\to j}$ and $C_j = C_{1\to j} \ot C'_j$ for $j=1,\dots,m$,
  as well as unitaries $U_{1\to j} \colon A_{1\to j} \to C_{1\to j}$ for $j=1,\dots,m$ and $U'_{A_2\dots{}A_n\to C'_1\dots{}C'_m}$,
  such that
  \[ U_{A_1\dots{}A_n\to C_1\dots{}C_m} = \bigotimes_{j=1}^n U_{1\to j} \ot U'_{A_2\dots{}A_n\to C'_1\dots{}C'_m}. \]
  Moreover, $U'$ is again a minimally $I_3$-scrambling MIMO unitary.
\end{lem}

\cref{thm:mimoUnitaries} now follows by applying \cref{lem:multipartite one to many} inductively to $A_1$, $A_2$, etc.
After $n$ steps, there are no $A$-systems left.
Since the residual operator $U'$ is a unitary, the corresponding $C'_j$ likewise have to be trivial.
We thus obtain the desired normal form.
To see that, conversely, any MIMO unitary of the given normal form is minimally $I_3$-scrambling follows directly from the corresponding statement in \cref{thm:unitaries}, applied to the bipartitions $A_i, A_i^c$ and $C_j, C_j^c$.
This concludes the proof of \cref{thm:mimoUnitaries}.

\subsection*{Maximal scrambling}
On the other end, we define a \emph{maximally $I_3$-scrambling} MIMO unitary as one that satisfies
\begin{equation*}
  I_3(A_i;A_i^c;C_j) = -2 \log \min\{\abs{A_i},\abs{A_i^c}, \abs{C_j}, \abs{C_j^c}\}
\end{equation*}
for all $i,j$.
Applying \cref{prp:depolarizing} to the bipartition $A_i , A_i^c, C_j, C_j^c$, we conclude that the residual channels $\mathcal N_{A_i\to C_j}$ are completely depolarizing whenever $A_i$ or $C_j$ is the smallest system (e.g., if all systems have the same dimension, as in a typical many-body scenario).
We note that if the average OTO correlators between $A_i,C_j$ and $A_i,C_j^c$ are minimal for each $i$ and $j$, then the MIMO unitary is maximally $I_3$-scrambling.

By an explicit construction similar to that of \cref{eq:odd scrambling}, we can establish that maximally $I_3$-scrambling MIMO unitaries exist for arbitrarily large values of $d$.

\begin{prp}
\label{prp:mimoex}
  Let $A_1=\dots=A_n=C_1=\dots=C_n=\C^d$, where $d>n+1$ is a prime.
  Let $M_n$ be the following $n \times n$ matrix,
  \begin{equation}
  \label{eq:mimoex}
  M_n = I_n + E_n = \begin{bmatrix}
    2 & 1 & \dots & 1 \\
    1 & 2 & \dots & 1 \\
      &   & \ddots \\
    1 & 1 & \dots & 2
  \end{bmatrix},
  \end{equation}
  where $I_n$ is the identity matrix and $E_n$ the matrix of ones.
  Then $U_{d,n} \ket{\vec x} = \ket{M_n \vec x}$ defines a maximally $I_3$-scrambling MIMO unitary.
  Here we write $\ket{\vec x} = \ket{x_1}\dots\ket{x_n}$, and all arithmetic is modulo $d$.
\end{prp}
We prove this by showing that the following three criteria on a matrix $M$ are together sufficient to ensure that $U_M \ket{\vec x} = \ket{M \vec x}$ is maximally $I_3$-scrambling:
\begin{enumerate}
\item $M$ is an invertible matrix modulo $d$.
\item If we replace any row of $M$ by any elementary row (i.e., a row with all 0's except for a single entry occupied by a 1) then the resulting matrix is still invertible modulo $d$.
\item All entries of $M$ are invertible modulo $d$.
\end{enumerate}
We then show that $M_n$ defined in~\eqref{eq:mimoex} satisfies these conditions when $d > n+1$ and is prime. The detailed proof is given in \cref{app:mimo}. It is an interesting open question to determine sufficient and necessary conditions on the dimensions for maximally $I_3$-scrambling MIMO unitaries to exist~\cite{goyeneche2015absolutely}.

\acknowledgments
We thank Sepehr Nezami and Beni Yoshida for helpful discussions.
The authors gratefully acknowledge support from the Simons Foundation, including the It from Qubit Collaboration, as well as CIFAR and the Air Force Office of Scientific Research. DD acknowledges funding by a Stanford Graduate Fellowship.

\appendix
\section{Nonrobustness in the approximately minimal case}\label{app:minimal}

In this appendix we prove \cref{lem:nonRobustZero}, restated here for convenience:

\lemmaNonRobustZero{nonRobustZero}{}{\tag{\ref{eq:nonRobustZero state dist}}}{\tag{\ref{eq:nonRobustZero diamond dist}}}
\begin{proof}
  Recall that $U_d$ is given by
  \begin{equation}
  \label{eq:counterUrepeated}
  U_d \ket a \ket b = \begin{cases}
    U_S \ket{a} \ket{b} & 0\leq a, b < d_S\\
    \ket{a} \ket{b} & \text{otherwise}
  \end{cases}.
  \end{equation}
  We prove the first statement.
  The Choi state $\rho_d=\ket{U_d}\!\!\bra{U_d}$ of $U_d$ is given by
  \begin{align}
  	\ket{U_d}
  &= \frac 1d U_{A'B'\to CD,d} \left( \sum_{a,b< d_S} + \sum_{a\geq d_S \lor b\geq d_S} \right) \ket{aa}_{AA'} \ket{bb}_{BB'} \nonumber \\
  &= \frac {d_S} d \ket{U_S}_{A_SB_SC_SD_S} + \frac 1d \sum_{a\geq d_S \lor b\geq d_S} \ket{aa}_{AC} \ket{bb}_{BD} \label{eq:counterchoi}
  \end{align}
  where we write $\ket{U_S}$ for the Choi state of the maximally $I_3$-scrambling unitary $U_S$.
  On the other hand,
  \[
    \ket{\Phi^+_{AC}} \ket{\Phi^+_{BD}}
  = \frac {d_S} d \ket{\Phi^+_{A_SC_S}} \ket{\Phi^+_{B_SD_S}} + \frac 1d \sum_{a\geq d_S \lor b\geq d_S} \ket{aa}_{AC} \ket{bb}_{BD}.
  \]
  Hence, for small $d_S$, the overlap between the two Choi states is given by
  \[
    \abs{\braket{\Phi^+_{AC} \ot \Phi^+_{BD} | U_d}}
  = \Bigl\lvert \frac {d_S^2} {d^2} \braket{\Phi^+_{A_SC_S} \ot \Phi^+_{B_SD_S} | U_S} + \frac {d^2 - d_S^2} {d^2} \Bigr\rvert
  \geq 1 - 2\frac {d_S^2} {d^2}
  \]
  Using the relationship between trace distance and overlap of pure states~\cite{wilde2013quantum},
  \begin{align*}
    \norm{\rho_{ABCD,d} - \Phi^+_{AC} \ot \Phi^+_{BD}}_1
  = 2 \sqrt{1 - \abs{\braket{\Phi^+_{AC} \ot \Phi^+_{BD} | U_d}}^2}
  \leq 4\frac {d_S} d.
  \end{align*}
  We have thus established~\eqref{eq:nonRobustZero state dist}.

	We now prove the second statement.
	Let $U_0$ be a minimally $I_3$-scrambling unitary.
	By Theorem \ref{thm:unitaries}, we can write
	\begin{equation}
	\label{eq:nofo mini3}
	  U_0 = U_{A_L \to C_L} \ot U_{A_R \to D_L} \ot U_{B_L \to C_R} \ot U_{B_R \to D_R},
	\end{equation}
	where $A = A_L \ot A_R$ and similarly for $B,C,D$.
  Without loss of generality, $\abs{C_R} \ge \abs{C}^{1/2}$. Otherwise, switch the roles of $A,B$ in the following.
  We consider a state of the form
  \begin{equation*}
    \sigma_{AB} = \sigma_{A_S} \ot \tau_B,
  \end{equation*}
  where $\sigma_{A_S}$ is an arbitrary state on $A_S \subseteq A = A_L \ot A_R$.
  We will show that $U_d$ and $U_0$ lead to reduced density matrices on $C$ with markedly different entropies, implying that $U_d,U_0$ are well-distinguishable.
  It is clear from~\eqref{eq:counterUrepeated} and the form of $\sigma_{AB}$ that $\sigma_C = \tr_D[U_d \sigma_{AB} U_d^\dagger]$ is supported on the subspace $C_S$, hence
  \[ S(\sigma_C) \leq \log d_S. \]
  On the other hand, using~\eqref{eq:nofo mini3} we can compute the second reduced state as
  \begin{align*}
    \sigma'_C
  &= \tr_D[U_0 \sigma_{AB} U_0^\dagger]
  = \bigl( U_{A_L \to C_L} \ot U_{B_L \to C_R}\bigr) \tr_{A_RB_R}[\sigma_{AB}] \bigl( U_{A_L \to C_L} \ot U_{B_L \to C_R} \bigr)^\dagger \\
  &= U_{A_L \to C_L} \tr_{A_R}[\sigma_{A_S}] U_{A_L \to C_L} \ot \tau_{C_R},
  \end{align*}
  and hence that
  \[ S(\sigma'_C) \geq \log C_R \geq \frac12 \log d. \]
  Thus, using the Fannes-Audenaert inequality~\eqref{eq:fannes audenaert},
	\begin{equation*}
	  \frac12 \log d - \log d_S \leq \abs{S(\sigma_C) - S(\sigma'_C)} \leq \frac12 \lVert \sigma_C - \sigma'_C \rVert_1 \log d + 1,
	\end{equation*}
	from which it follows that
	\begin{equation*}
	  \lVert \sigma_C - \sigma'_C \rVert_1 \geq 1 - \frac {2 + 2 \log d_S} {\log d}.
	\end{equation*}
	Hence, we can bound the trace distance between the output states using monotonicity, which in turn bounds the diamond norm~\eqref{eq:diamond}:
	\[ \lVert U - U_0 \rVert_\diamond
	\geq \lVert U_d \sigma_{AB} U_d^\dagger - U_0 \sigma_{AB} U_0^\dagger \rVert_1
	\geq \lVert \sigma_C - \sigma'_C \rVert_1
	\geq 1 - \frac {2 + 2\log d_S} {\log d}.
  \]
  This establishes~\eqref{eq:nonRobustZero diamond dist}.
\end{proof}

\section{Nonrobustness in the approximately maximal case}\label{app:maximal}

In this appendix we prove \cref{lem:nonRobustMax}, restated again for convenience.

\lemmaNonRobustMax{nonRobustMax}{}
\begin{proof}
  %

  We start with the formula in~\eqref{eq:counterchoi} for the Choi state of $U_d$, which can be written as
  \begin{align*}
    \ket{U_d} 
    &= \frac {d_S}d \ket{U_S}_{A_SB_SC_SD_S}
    + \frac {d_0}d \ket{\Phi^+_{A_0C_0}} \ot \ket{\Phi^+_{B_0D_0}} \\
    &+ \frac {\sqrt{d_S d_0}}d \ket{\Phi^+_{A_SC_S}} \ot \ket{\Phi^+_{B_0D_0}}
    + \frac {\sqrt{d_S d_0}}d \ket{\Phi^+_{A_0C_0}} \ot \ket{\Phi^+_{B_SD_S}}
  \end{align*}
  where $A=A_S\op A_0$ etc.\, with $A_S$, $B_S$, etc.\ the $d_S$-dimensional subspaces on which the maximally $I_3$-scrambling unitary $U_S$ acts, and $\ket{U_S}$ the Choi state of the latter.

  We first compute the reduced density matrix $\rho_{AD,d}$.
  There are no cross-terms, hence
  \[ \rho_{AD,d}
  = \frac {d_S^2} {d^2} \tau_{A_SD_S}
  + \frac {d_0^2} {d^2} \tau_{A_0D_0}
  + \frac {d_0 d_S} {d^2} \tau_{A_SD_0}
  + \frac {d_0 d_S} {d^2} \tau_{A_0D_S}
  = \tau_{AD} \]
  as desired. Here, we have used that $U_S$ is maximally $I_3$-scrambling and hence its reduced state on $A_SD_S$ is maximally mixed.

  We now compute the reduced density matrix $\rho_{AC,d}$.
  For this, we split the matrix into blocks according to the decomposition $AC = A_SC_S \op A_0C_0 \op A_SC_0 \op A_0C_S$.
  Then there are four nonzero blocks,
  \begin{equation}
  \label{eq:AC blocks}
	  \rho_{AC,d}=\left[
	    \begin{array}{c|c|c|c}
	      \rho_{SS}  & \rho_{S0} & 0 & 0 \\ \hline
	      \rho_{S0}^\dagger & \rho_{00} & 0 & 0 \\ \hline
	      0 & 0 & 0 & 0 \\ \hline
	      0 & 0 & 0 & 0
	    \end{array}\right],
  \end{equation}
  where
  \begin{equation}
  \label{eq:AC pieces}
  \begin{aligned}
    \rho_{SS} &= \frac {d_S^2} {d^2} \tau_{A_SC_S} + \frac {d_0d_S} {d^2} \Phi^+_{A_SC_S}, \\
    \rho_{00} &= \frac {d_0d_S} {d^2} \Phi^+_{A_0C_0} + \frac {d_0^2} {d^2} \Phi^+_{A_0C_0} = \frac {d_0} d \Phi^+_{A_0C_0}, \\
    \rho_{S0} &= \frac {\sqrt{d_0d_S}} {d^2} \ket{\Psi_{A_SC_S}}\!\!\bra{\Phi^+_{A_0C_0}},
  \end{aligned}
  \end{equation}
  where we have introduced
  \begin{align*}
    \ket{\Psi_{A_SC_S}} &= \ket{\theta_{A_SC_S}} + d_0 \ket{\Phi^+_{A_SC_S}},
  \end{align*}
  where $\ket{\theta_{A_SC_S}} = d_S \braket{\Phi^+_{B_SD_S} | U_{S,A_SB_SC_SD_S}}$.
  This is a unit vector:
  \[ \braket{\theta|\theta} = d_S^2 \tr \Phi^+_{B_SD_S} \tau_{B_SD_S} = 1, \]
  since $U_S$ is maximally $I_3$-scrambling and so its Choi state on $B_SD_S$ is maximally mixed.
  It follows that
  \[ \lVert \rho_{S0} \rVert_1
  = \tr \sqrt{\rho_{S0} \rho_{S0}^\dagger}
  = \frac {\sqrt{d_0d_S}} {d^2} \lVert \Psi_{A_SC_S} \rVert
  \leq \frac {\sqrt{d_0d_S}} {d^2} (1+d_0)
  \leq 2 \frac {d_0} d
  \]
  Therefore, using $\tau_{AC} = \frac {d_S^2} {d^2} \tau_{A_SC_S} + \frac {d^2-d_S^2} {d^2} \tau'$, where $\tau'$ is a maximally mixed state on the complement of $A_SC_S$,
  \begin{align*}
    \lVert \rho_{AC,d} - \tau_{AC} \rVert_1
    &= \lVert \frac {d_0d_S} {d^2} \Phi^+_{A_SC_S} + \rho_{00} + \rho_{S0} + \rho^\dagger_{S0} - \frac {d^2-d_S^2} {d^2}\tau' \rVert_1 \\
    &\leq \frac {d_0d_S} {d^2} + \frac {d_0} d + 4 \frac {d_0} d + \frac {d^2-d_S^2} {d^2}
    \leq 8 \frac {d_0}d.
  \end{align*}

	At last, we show that the residual channel $\mathcal N_{A\to C}$ for $U_d$ is bounded away from the completely depolarizing channel $\Delta_{A\to C}$ in the diamond norm.
	For this, it suffices to compare their action on a state orthogonal to the scrambling subspace $A_S$, so that $\mathcal N_{A\to C}$ acts by the identity.
	The $d$-th computational basis state $\ket{d-1}$ is such a state:
	\[ \bigl\lVert \mathcal N_{A\to C} - \Delta_{A\to C} \bigr\rVert_\diamond \geq \bigl\lVert \ket{d-1}\!\!\bra{d-1}_C - \tau_C \bigr\rVert_1 = 2 - \frac2d. \qedhere \]
\end{proof}

\section{Calculation of the R\'enyi-2 tripartite information}\label{app:renyi}

In this appendix we verify~\eqref{eq:renyi tripartite}, the lower bound for the R\'enyi-2 tripartite information of the unitary $U_d$ defined in~\eqref{eq:counterU}.
Let $\rho_{ABCD,d}$ denote its Choi state.
In \cref{lem:nonRobustMax}, we have shown that $\rho_{AD,d}$ is maximally mixed.
Hence the R\'enyi-2 tripartite information~\eqref{eq:renyi3 defn} reduces to
\[ I_3^{(2)} = -S_2(AC) = \log \tr \rho_{AC,d}^2. \]
Now, it follows from~\eqref{eq:AC blocks} that
\[
  \rho_{AC,d}^2 =
	\left[\begin{array}{c|c}
    \rho_{SS}^2 + \rho_{S0}\rho_{S0}^\dagger & * \\ \hline
    * & \rho_{00}^2 + \rho_{S0}^\dagger\rho_{S0}
  \end{array}\right],
\]
where we omitted zero rows and did not specify the off-diagonal blocks, which are irrelevant to our calculation.
Using~\eqref{eq:AC pieces}, we find
\begin{align*}
  \tr \rho_{AC,d}^2
\geq \tr \rho_{SS}^2 + \tr \rho_{00}^2
= \frac {d_S^2} {d^4} + 2 \frac {d_0d_S}{d^4} + \frac {d_0^2 d_S^2} {d^4} + \frac {d_0^2} {d^2}
\geq \frac {d_0^2 d_S^2} {d^4}
= \frac 1 {d^2} \frac {d_S^2} {d^2} d_0^2.
\end{align*}
Hence, if we choose $d_0 \sim \sqrt[4]d$ then $\log d_0 \sim \frac14\log d$, thus
\begin{align*}
  I_3^{(2)}
  &= \log \tr \rho_{AC,d}^2 \geq - 2 \log d + 2 \log (1 - \frac {d_0} d) + \log d_0^2 \gtrsim - \frac32 \log d.
\end{align*}

\section{Maximal scrambling and typical inputs}\label{app:typical}

In \cref{prp:depolarizing} we found that the residual channel $\mathcal{N}_{A \to C}$ for maximally mixed input on $B$ is completely depolarizing.
In other words, its Choi state is maximally mixed, $\mathcal{N}_{A' \to C}[\Phi^+_{AA'}] = \tau_A \ot \tau_C$.
Under certain conditions this is approximately true also for \emph{typical} input states on $B$:

\begin{prp}
  \label{prp:randomB}
  Let $U_{AB\to CD}$ be a maximally $I_3$-scrambling unitary and $\sigma_B$ a Haar-random pure state.
  Let $\widetilde{\mathcal N}_{A\to C}[\sigma_A] = \tr_D[U (\sigma_A \ot \sigma_B) U^\dagger]$ denote the corresponding residual channel from $A$ to $C$, and $\widetilde\rho_{AC} := \widetilde{\mathcal N}_{A'\to C}[\Phi^+_{AA'}]$ its Choi state.
  Then,
  \begin{equation*}
    \Pr(\norm{\widetilde\rho_{AC} - \tau_{AC}}_1 \leq \varepsilon) \ge 1- \frac{\abs{A}\abs{C}}{\varepsilon^2 \abs{D}}
   \end{equation*}
\end{prp}
\begin{proof}
  Let us write $\rho_{ABCD}$ for the Choi state of $U_{AB\to CD}$.
  For a Haar-random pure state, $\EE[\sigma_B] = \tau_B$.
  Hence, the average Choi state is maximally mixed, $\EE[\widetilde\rho_{AC}] = \rho_{AC} = \tau_{AC}$.

  We now bound the mean square deviation.
  For this, let $\norm{X}_2 := \sqrt{\tr X^\dagger X}$ denote the 2-norm.
  Then:
  \[ \EE[\lVert \widetilde\rho_{AC} - \tau_{AC} \rVert_2^2] = \EE[\tr \widetilde\rho_{AC}^2] - \tr \tau_{AC}^2. \]
  We calculate the first term using the swap trick:
  \begin{align*}
    \tr \widetilde\rho_{AC}^2 = \tr (\widetilde\rho_{AC} \ot \widetilde\rho_{AC}) F_{AC} = \tr U^{\ot 2}_{A'B\to CD} (\Phi^{+\ot2}_{AA'} \ot \sigma_B^{\ot2}) U^{\dagger\ot2}_{A'B\to CD} F_{AC},
  \end{align*}
  where $F_{AC}$ denotes the swap operator that exchanges the two copies of $AC$.
  The second moment of a Haar-random state is given by $\EE[\sigma_B^{\ot2}] = \frac 1 {\abs B (\abs B+1)} (I + F_B)$ where $I$ is the identity and $F_B$ the swap operator on the two copies of $B$ .
  Thus:
  \begin{align*}
    \EE[\tr \widetilde\rho_{AC}^2]
  &= \frac 1 {\abs B (\abs B+1)} \tr U^{\ot 2}_{A'B\to CD} \Phi^{+\ot2}_{AA'} U^{\dagger\ot2}_{A'B\to CD} F_{AC} \\
  &+ \frac 1 {\abs B (\abs B+1)} \tr U^{\ot 2}_{A'B\to CD} \Phi^{+\ot2}_{AA'} F_B U^{\dagger\ot2}_{A'B\to CD} F_{AC}.
  \end{align*}
  The first term can be bounded as
  \begin{align*}
	&\quad	\frac 1 {\abs B (\abs B+1)} \tr U^{\ot 2}_{A'B\to CD} \Phi^{+\ot2}_{AA'} U^{\dagger\ot2}_{A'B\to CD} F_{AC} \\
	&= \frac {\abs B^2} {\abs B (\abs B+1)} \tr U^{\ot 2}_{A'B\to CD} (\Phi^{+\ot2}_{AA'} \ot \tau_B^{\ot 2}) U^{\dagger\ot2}_{A'B\to CD} F_{AC} \\
	&= \frac {\abs B^2} {\abs B (\abs B+1)} \tr \tau_{AC}^2
	\leq \tr \tau_{AC}^2,
  \end{align*}
  where the last equality follows since the Choi state of $U$ is maximally mixed on $AC$.
  For the second term, we compute
  \begin{align*}
  &\quad \frac 1 {\abs B (\abs B+1)} \tr U^{\ot 2}_{A'B\to CD} \Phi^{+\ot2}_{AA'} F_B U^{\dagger\ot2}_{A'B\to CD} F_{AC} \\
  &= \frac {\abs B^2} {\abs B (\abs B+1)} \tr U^{\ot 2}_{A'B'\to CD} (\Phi^{+\ot2}_{AA'} \ot \tau^{\ot2}_{B'}) F_{B'} U^{\dagger\ot2}_{A'B'\to CD} F_{AC} \\
  &=\frac {\abs B^2} {\abs B (\abs B+1)} \tr U^{\ot 2}_{A'B'\to CD} (\Phi^{+\ot2}_{AA'} \ot \Phi^{+\ot 2}_{BB'}) U^{\dagger\ot2}_{A'B'\to CD} F_{ABC} \\
  &=\frac {\abs B^2} {\abs B (\abs B+1)} \tr \rho_{ABC}^2
  \leq \tr \rho_D^2 = \frac 1 {\abs D}.
  \end{align*}
  In the first step, we have relabeled $B$ to $B$' and inserted two copies of the maximally mixed state $\tau_{B'}$; 
  in the second, we have extended the maximally mixed states to maximally entangled states $\Phi^+_{BB'}$ and teleported the swap operator from the $B'$ systems to the $B$ systems;
  in the third step, we have recognized the Choi state of $U$ and undone the swap trick;
  and in the last we have used that $\rho_D$ is maximally mixed.
  Together, we obtain the following bound on the mean square deviation:
  \[ \EE[\lVert \widetilde\rho_{AC} - \tau_{AC} \rVert_2^2] \leq \frac 1 {\abs D}. \]
  By the Cauchy-Schwarz inequality, $\norm{X}^2_1 \leq \abs A \abs C \, \norm{X}^2_2$, we get
  \[ \EE[\lVert \widetilde\rho_{AC} - \tau_{AC} \rVert_1^2] \leq \frac {\abs A \abs C} {\abs D}. \]
  Now Markov's inequality gives
  \[ \Pr(\lVert \widetilde\rho_{AC} - \tau_{AC} \rVert_1 \geq \varepsilon)
  = \Pr(\lVert \widetilde\rho_{AC} - \tau_{AC} \rVert_1^2 \geq \varepsilon^2)
  \leq \frac {\EE[\lVert \widetilde\rho_{AC} - \tau_{AC} \rVert_1^2]} {\varepsilon^2}
  \leq \frac {\abs A\abs C}{\varepsilon^2 \abs D},
  \]
  and we obtain the desired bound:
  \[\Pr(\lVert \widetilde\rho_{AC} - \tau_{AC} \rVert_1 \leq \varepsilon) \ge 1- \frac {\abs A\abs C}{\varepsilon^2 \abs D}.\qedhere \]
\end{proof}

The fact that we need $\abs A \abs C \ll \abs D$ is intuitive:
For any realization of the random pure state $\sigma_B$, the state $\widetilde\rho_{ACD} = U_{A'B\to CD} \ket{\Phi^+_{AA'}} \ot \ket{\sigma_B}$ is a purification of $\rho_{AC}$.
Thus, if $\rho_{AC}$ is to be maximally mixed then we clearly need that $\abs A \abs C \leq \abs D$, since otherwise the Schmidt rank cannot be $\abs A \abs C$.

\smallskip

One natural scenario to apply \cref{prp:randomB} is to the toy model of black hole evaporation discussed on p.~\pageref{black holes} (with $D$ and $C$ interchanged). If $A$ is small (e.g., a qubit) and the initial black hole $B$ is in a typical pure state, the Hawking radiation emitted at later times $D$ is decoupled from $A$ if $D$ is much smaller than the post-evaporation black hole $C$~\cite{hayden2007black}. The only assumption necessary about the dynamics is that the black hole be maximally $I_3$-scrambling.

Another natural scenario to apply~\cref{prp:randomB} is in the context of maximally $I_3$-scrambling MIMO unitaries as discussed in \cref{sec:multipartite}.
Here, $\abs{A_i}\abs{C_j}$ is usually much smaller than $\abs{C_j^c}$.
Hence, if we input a random pure state into $A_i^c$ and half of a maximally entangled state into $A_i$, then with high probability the reduced state on $A_i C_j$ is close to being maximally mixed.
We can make a even stronger statement by demanding
\begin{equation*}
  \norm{\widetilde\rho_{A_iC_j} - \tau_{A_iC_j}}_2 \le \frac{1}{\abs{A_i}\abs{C_j}},
\end{equation*}
which by~\cite{gurvits2002separable} would imply that $\widetilde\rho_{A_iC_j}$ is separable.
By Choi-Jamio\l{}kowski, this means $\mathcal{N}_{A_i' \to C_j}^{\sigma_{A_i^c}}$ is entanglement-breaking.
Using \cref{prp:randomB} the probability of this is at least $1 - {\abs{A_i}^2 \abs{C_j}^2}/{\abs{C_j^c}}$. In the case where all systems are of size $d$,
\begin{equation*}
  \frac{\abs{A_i}^2 \abs{C_j}^2}{\abs{C_j^c}} = \frac{1}{d^{n-5}},
\end{equation*}
which vanishes for large $n$ or $d$.

\section{Existence of maximally scrambling MIMO unitaries}\label{app:mimo}
In this appendix we prove \cref{prp:mimoex}.
As discussed in \cref{sec:multipartite}, we first consider the case where $M_n$ is replaced by an arbitrary $n \times n$ matrix $M$ and identify sufficient conditions for the corresponding unitary $U_M \ket{\vec x}=\ket{M\vec x}$ to be maximally $I_3$-scrambling.
First, it is clear that $U_M$ is unitary if and only if $M$ is invertible modulo $d$.
We then consider the Choi state of $U_M$,
\[ \rho_{AC} = \frac 1 {d^n} \sum_{\vec x, \vec y} \ket{\vec x}\!\!\bra{\vec y}_A \otimes \ket{M\vec x}\!\!\bra{M\vec y}_C, \]
where we write $A=A_1\dots{}A_n$ and $C=C_1\dots{}C_n$.
We now compute the reduced state $\rho_{A_i^cC_j}$.
The partial trace over $A_i$ forces $x_i=y_i$, and the partial trace over $C_j^c$ forces $M\vec x=M\vec y$, except for the $j$-th entry.
Assuming that matrix we obtain by replacing the $j$-th row of $M$ with the elementary row $e_i$ is invertible modulo $d$, this implies that $\vec x = \vec y$.
Hence,
\[ \rho_{A_i^cC_j} = \frac 1 {d^n} \sum_{\vec x} \ket{\vec x'}\!\!\bra{\vec x'}_{A_i^c} \ot \ket{(M\vec x)_j}\!\!\bra{(M\vec x)_j}_{C_j}, \]
where $\vec x'$ is $\vec x$ with the $i$-th entry omitted and $(M\vec x)_j$ denotes the $j$-th entry of $M\vec x$.
First summing over $\vec x'$ and then over all options for $x_i$, we get
\[ \rho_{A_i^cC_j} = \frac 1 {d^{n-1}} \sum_{\vec x'} \ket{\vec x'}\!\!\bra{\vec x'}_{A_i^c} \ot \frac1d \sum_{x_i} \ket{(M\vec x)_j}\!\!\bra{(M\vec x)_j}_{C_j}. \]
Assuming the matrix element $M_{ji}$ is invertible modulo $d$, the right-hand side sum is over all basis states, for any fixed choice of $\vec x'$.
Thus:
\[ \rho_{A_i^cC_j} = \frac 1 {d^{n-1}} \sum_{\vec x'} \ket{\vec x'}\!\!\bra{\vec x'}_{A_i^c} \ot \tau_{C_j} = \tau_{A_i^cC_j}. \]
If we replace $i$ by some $k\neq i$ then we also find that
\[ \rho_{A_iC_j} = \tr_{A_k^c \setminus A_i}[\rho_{A_k^cC_j}] = \tau_{A_iC_j}. \]
Together, we obtain that
\[ I_3(A_i ; A_i^c ; C_j) = -I(A_i ; A_i^c | C_j)
= -2\log d,
\]
as desired.
Hence, it is sufficient for $M$ to satisfy the following three criteria so that $U_M$ is maximally $I_3$-scrambling:
\begin{enumerate}
\item $M$ is an invertible matrix modulo $d$.
\item If we replace any row of $M$ by any elementary row then the resulting matrix is still invertible modulo $d$.
\item All entries of $M$ are invertible modulo $d$.
\end{enumerate}
Now we show that $M_n$ defined in~\eqref{eq:mimoex} satisfies these conditions when $d > n+1$ and is prime.
Recall that
\begin{equation*}
  M_n = \begin{bmatrix}
    2 & 1 & \dots & 1 \\
    1 & 2 & \dots & 1 \\
      &   & \ddots \\
    1 & 1 & \dots & 2
  \end{bmatrix}.
\end{equation*}
The third condition is obvious, since both $1$ and $2$ are invertible modulo $d$.
For the first condition, we note that $n+1$ is invertible modulo $d$.
Hence the following matrix is well-defined and easily checked to be the inverse of $M_n$:
\begin{equation*}
  M_n^{-1} = -(n+1)^{-1} \begin{bmatrix}
    -n & 1 & \dots & 1 \\
    1 & -n & \dots & 1 \\
      &   & \ddots \\
    1 & 1 & \dots & -n
  \end{bmatrix}
\end{equation*}
It remains to verify the second criterion.
If we replace the $j$-th row by an elementary row $e_i$, we obtain a matrix of the form
\begin{equation*}
  N_n = \begin{bmatrix}
    2 & 1 & \dots & 1 & 1\\
    1 & 2 & \dots & 1 & 1\\
      &   & \dots \\
    0 & \dots & 1 & \dots &0 \\
      &   & \dots \\
    1 & 1 & \dots & 1 & 2
  \end{bmatrix}.
\end{equation*}
We can calculate the determinant by cofactor expanding along the elementary row.
If $i=j$ then we obtain that $\det N_n = \pm \det M_{n-1}$, which is nonzero by the preceding.
Otherwise, if $i\neq j$ then find that $\det N_n$ is up to sign equal to the determinant of the following $(n-1)\times(n-1)$ matrix,
\begin{equation*}
  N'_{n-1} =
  \begin{bmatrix}
    2 & 1 & \dots & 1 & 1\\
    1 & 2 & \dots & 1 & 1\\
      &   & \dots \\
    1 & \dots & 1 & \dots &1 \\
      &   & \dots \\
    1 & 1 & \dots & 1 & 2
  \end{bmatrix},
\end{equation*}
which looks like $M_{n-1}$ except that a $2$ is replaced by a $1$.
We can use determinant-preserving row operations to reduce this matrix to
\begin{equation*}
  \begin{bmatrix}
    1 & 0 & \dots & 0 & 0\\
    0 & 1 & \dots & 0 & 0\\
      &   & \dots \\
    1 & \dots & 1 & \dots &1 \\
      &   & \dots \\
    0 & 0 & \dots & 0 & 1
  \end{bmatrix},
\end{equation*}
which has determinant one.
This concludes the proof of \cref{prp:mimoex}.


\bibliographystyle{JHEP}
\bibliography{paper}

\end{document}